\documentclass[a4paper,twoside]{article}

\usepackage{multirow}

\usepackage{epsfig}
\usepackage{calc}
\usepackage{amssymb}
\usepackage{amstext}
\usepackage{amsmath}
\usepackage{amsthm}
\usepackage{multicol}
\usepackage{pslatex}
\usepackage{apalike}
\usepackage[ruled, noend,  linesnumbered]{algorithm2e} 

\usepackage{hyperref}
\usepackage{graphicx,subcaption}
\usepackage{caption}

\usepackage[final]{listings}
\usepackage{proof}
\newtheorem{thm}{Theorem}

\theoremstyle{definition}
\newtheorem{definition}{Definition}[section]
\usepackage{paralist}
\usepackage{setspace}

\makeatletter
\def\thickhline{%
  \noalign{\ifnum0=`}\fi\hrule \@height \thickarrayrulewidth \futurelet
   \reserved@a\@xthickhline}
\def\@xthickhline{\ifx\reserved@a\thickhline
               \vskip\doublerulesep
               \vskip-\thickarrayrulewidth
             \fi
      \ifnum0=`{\fi}}
\makeatother

\newlength{\thickarrayrulewidth}
\lstdefinestyle{sparql}{%
    morekeywords={SELECT,OPTIONAL,FROM,DISTINCT,a,WHERE,FILTER,GROUP,ORDER,LIMIT,BY,IN,AS},
    emph={r,pub,aairObject,verb,person,bday,s,p,o}
}
\lstdefinestyle{turtle}{%
    morekeywords={a, @prefix},
    morecomment=[s][\textrm]{<}{>},
    morecomment=[s][\textit]{"}{"},
}
\lstdefinestyle{xml}{%
    morekeywords={a, @prefix},
    morecomment=[s][\textrm]{<}{>},
    morecomment=[s][\textit]{"}{"},
    keywordstyle=\color{cyan},
}

\lstdefinelanguage{XML}
{
  stringstyle=\color{black},
  identifierstyle=\color{darkblue},
  keywordstyle=\color{cyan},
}

\usepackage{SCITEPRESS} 
\begin{document}

\title{Compact Representations for Efficient Storage \\of Semantic Sensor Data}

\author{\authorname{Farah Karim\sup{1}\sup{2}, Maria-Esther Vidal\sup{1} and S{\"o}ren Auer\sup{1}}
\affiliation{\sup{1}Leibniz University of Hannover,
              Welfengarten 1B,  30167 Hannover, Germany}
\affiliation{\sup{2}Mirpur University of Science and Technology (MUST), Mirpur-10250 (AJK), Pakistan}
\email{\{karim, vidal, auer\}@l3s.de}
}


\abstract{Nowadays, there is a rapid increase in the number of sensor data generated by a wide variety of sensors and devices.
Data semantics facilitate information exchange, adaptability, and interoperability among several sensors and devices.
Sensor data and their meaning can be described using ontologies, e.g., the Semantic Sensor Network (SSN) Ontology.
Notwithstanding, semantically enriched, the size of semantic sensor data is substantially larger than raw sensor data. Moreover, some measurement values can be observed by sensors several times, and a huge number of repeated facts about sensor data can be produced.
We propose a \textit{compact} or \textit{factorized} representation of semantic sensor data, where repeated measurement values are described only once. 
Furthermore, these compact representations are able to enhance the storage and processing of semantic sensor data.
To scale up to large datasets, factorization based,  tabular representations are exploited to store and manage factorized semantic sensor data using Big Data technologies. We empirically study the effectiveness of a semantic sensor's proposed compact representations and their impact on query processing.
Additionally, we evaluate the effects of storing the proposed representations on diverse RDF implementations. 
Results suggest that the proposed compact representations empower the storage and query processing of sensor data over diverse RDF implementations, and up to two orders of magnitude can reduce query execution time.
}

\onecolumn \maketitle \normalsize \vfill

\section{\uppercase{Introduction}}
Internet of Things (IoT), cyber-physical systems, and sensor data applications are of paramount importance in our increasingly data-centric society and receive growing attention from the research community. RDF representations of IoT data are being generated~\cite{gaur2015smart,jabbar2017semantic} to add semantics to the data and turn the data into meaningful actions. 
% for providing the IoT applications with new capabilities, facilitate knowledge sharing and exchange, and richer experiences. 
The Semantic Sensor Network (SSN) Ontology~\cite{compton2012ssn} is a W3C standard to describe the sensor data, refer as semantic sensor data. The SSN Ontology consists of several classes and corresponding properties to describe the meaning of sensor data in terms of sensor capabilities, observations, and measured values in an RDF graph.
However, RDF representations generate an enormous amount of data; thus, efficient representations of sensor data are required. Furthermore, several sensor observations with the same measurement values generate RDF data redundancy, e.g., {\tt 13$^\circ$F} temperature observed by several sensors over the different timestamps. These data redundancies negatively impact the size of the semantic sensor data, hence the storage and processing of this data. Therefore, efficient representations of semantic sensor data are required to store and process large amounts of sensor data using different RDF implementations.
Rule-based \cite{joshi2013logical,meier2008towards,pichler2010redundancy} and binary \cite{alvarez2011compressed,bok2019provenance,FernandezMGPA13,pan2014ssp} compression techniques for RDF data effectively reduce the size of the data.
%, but these approaches require data decompression or customized engines to process the compressed data. 
%Furthermore, in order to scale-up to large RDF datasets, existing approaches 
%\cite{du2012hadooprdf,khadilkar2012jena,mami2016towards,nie2012efficient,papailiou2013h,punnoose2012rya,schatzle2012cascading,schatzle2013pigsparql,schatzle2014sempala} exploit 
Distributed and parallel processing frameworks for Big Data are exploited in several approaches~\cite{du2012hadooprdf,khadilkar2012jena,mami2016towards,nie2012efficient,papailiou2013h,punnoose2012rya,schatzle2012cascading,schatzle2013pigsparql}. %However, these approaches can be improved by using the RDF factorized representations. 
Moreover, column-oriented stores \cite{idreos2012monetdb,macnicol2004sybase,stonebraker2005c,zukowski2006super} apply column-wise compression techniques, and improve query performance by projecting the required columns. 
%These column stores are implemented based on the fully decomposed storage model~\cite{copeland1985decomposition} requiring more storage space. 
In the context of query processing, efficient SQL query processing techniques based on the factorization of the data are proposed in \cite{BakibayevKOZ13}. 
Despite these storage and processing techniques, the tremendously growing data requires efficient representations to facilitate the storage and processing.
%We present RDF factorization techniques tailored for the semantic sensor data to facilitate efficient storage and query processing. 
\\
\noindent
\textbf{Our Research Goal:}
We tackle the problem of efficiently representing semantic sensor data described using the SSN ontology. 
Our research goal is to generate compact representations where redundancies are removed. The proposed compact representations enhance the performance of the query engines by scaling up to large sensor data. We aim at determining the impact of the compact representations on semantic sensor data, and the effect of these representations on query processing.
\\
\noindent
\textbf{Approach:}
In this work, we propose the \textit{Compacting Semantic Sensor Data (CSSD)} approach for efficient storage and processing of semantic sensor data.
The \textit{CSSD} approach is based on factorizing the data and storing only a \textit{compact} or \textit{factorized} representation of semantic sensor data, where repeated values are represented only once. 
In addition, universal~\cite{ullman1984principles} and Class Template (CT) based tabular representations leveraging the columnar-oriented \emph{Parquet} storage format are utilized to scale up to even larger RDF datasets.

\begin{figure*}[t!]
\centering
\subfloat[Original {\bf RDF Graph}]{\includegraphics[width=.25\linewidth]{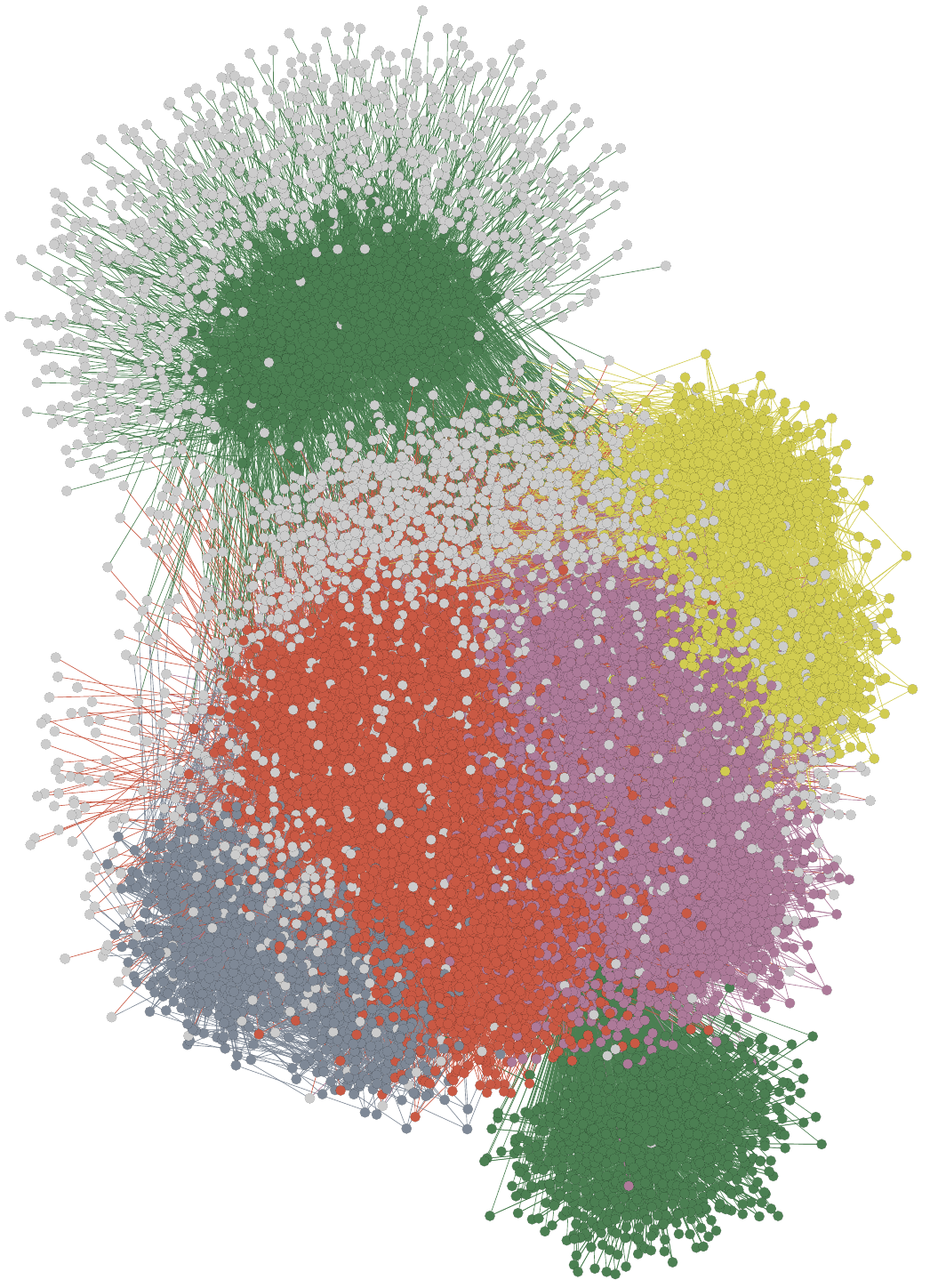}\label{fig:OriginalRDFGraph}} 
\subfloat[{\bf Statistics} of Original RDF Graph]{
	\begin{tabular}[b]{|c|c|c|}
		\hline
       \multirow{1}{*}{\textbf{S\#}} & \multicolumn{1}{c|}{\multirow{1}{*}{\textbf{Parameter}}} & \multirow{1}{*}{\textbf{Value}}  \\ \hline
        \multirow{1}{*}{$1$}&\multirow{1}{*}{Connected Components}&\multicolumn{1}{r|}{\multirow{1}{*}{$1.0$}} \\     \hline
	    \multirow{1}{*}{$2$}&\multirow{1}{*}{Network Centralization} &\multicolumn{1}{r|}{\multirow{1}{*}{$0.3$}} \\   
	       \hline
	    \multirow{1}{*}{\textbf{3}}&\multirow{1}{*}{\textbf{Avg. \# of Neighbors}} &\multicolumn{1}{r|}{\multirow{1}{*}{\textbf{6.4}}} \\     \hline
	    \multirow{1}{*}{$4$}&\multirow{1}{*}{Network Density} &\multicolumn{1}{r|}{\multirow{1}{*}{$0.0$}} \\   
	     \hline
	    \multirow{1}{*}{5}&\multirow{1}{*}{Multi-edge Node Pairs} &\multicolumn{1}{r|}{\multirow{1}{*}{5,149.0}} \\    
    \hline
	    \multirow{1}{*}{$6$}&\multirow{1}{*}{Network Heterogeneity} &\multicolumn{1}{r|}{\multirow{1}{*}{$11.1$}}\\   
	        \hline
	\end{tabular}
\label{fig:statsOriginal}} 
\\
\vspace{5pt}
\subfloat[{\bf NT per Value} vs Total]{\includegraphics[width=0.6\linewidth]{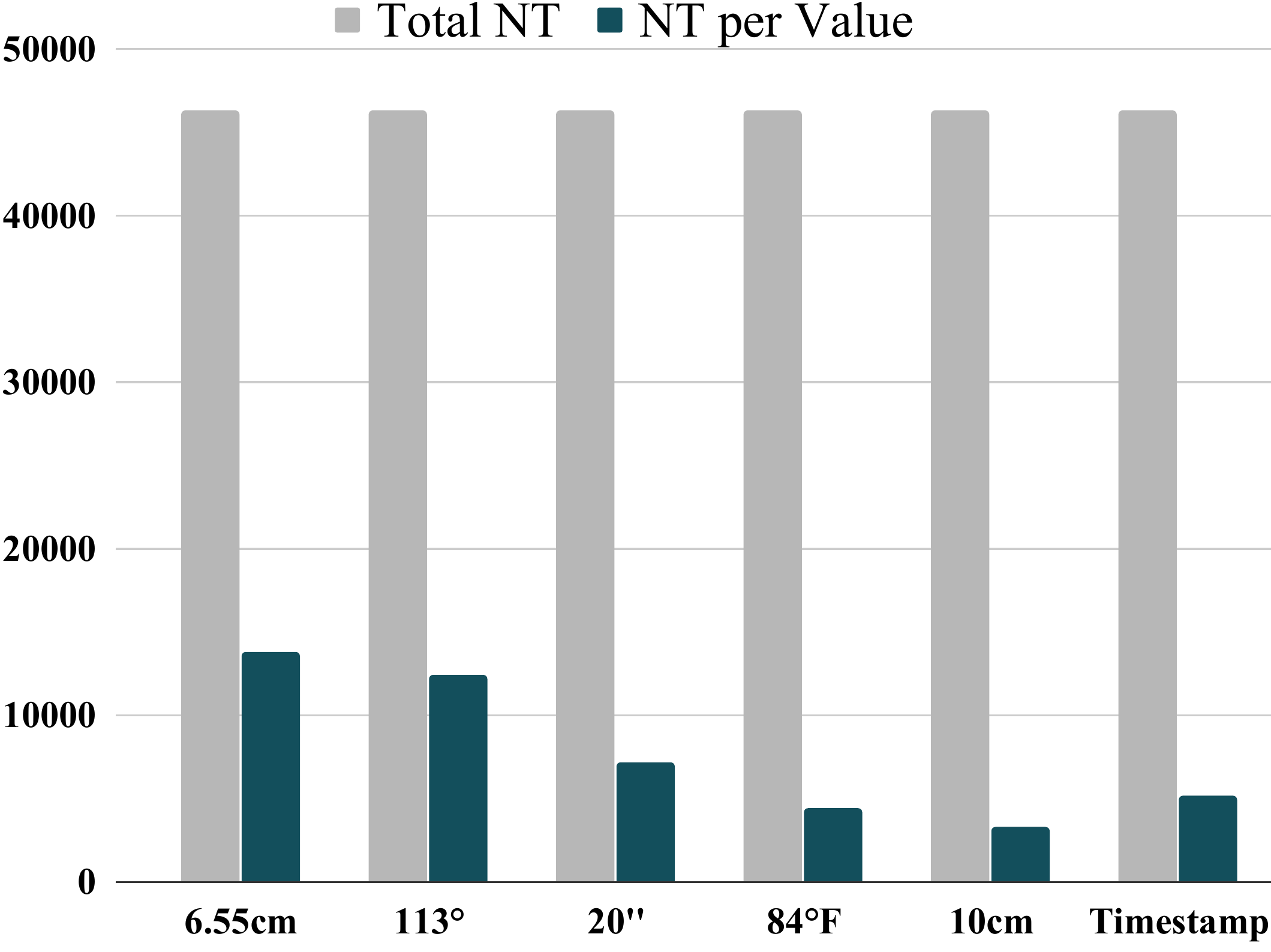}\label{fig:NTTotal}} 
\caption{{\bf Motivating Example.} (a) An RDF Graph, with the same color of nodes and edges, represents the RDF triples related to same values; (b) Statistics of the RDF graph; (c) Number of RDF triples(NT), associated with same value, to total RDF triples. % in the RDF graph in Figure~\ref{fig:OriginalRDFGraph}. 
The RDF graph and statistics are generated by {\tt Cytoscape tool}. (\url{http://www.cytoscape.org/}).}
\label{fig:motivatingExample}
\end{figure*}
The effectiveness of the proposed factorization techniques are empirically studied, as well as the impact of factorizing semantic sensor data on query processing using \emph{LinkedSensorData} benchmark~\cite{patni2010linked}. 
%The effects of storing factorized RDF data over diverse RDF implementations using state-of-the-art RDF and Big Data engines are evaluated. We study the effectiveness of the proposed compact representations over the \emph{LinkedSensorData} benchmark~\cite{patni2010linked}. 
%The \emph{LinkedSensorData} contains almost 2 billion RDF triples to describe more than 34 million weather observations collected from around 20,000 weather stations during blizzard and storm seasons in the United States since 2002. 
%The experiments are conducted over gradually increasing three \emph{LinkedSensorData} RDF graphs. 
The observed results demonstrate that the proposed factorization techniques are able to effectively reduce the size of semantic sensor data while all the encoded information is preserved, and improves query performance. 
%Moreover, the proposed techniques are able to reduce the query execution time by up to two orders of magnitude.
This article extends our previous work \cite{karim2017large}, where we introduce the factorization techniques for semantic sensor data to scale-up to large datasets. Here, we present techniques for efficient storage of semantic sensor data and conduct an extended analysis and evaluations of the \textit{CSSD} approach. In essence, we make the following contributions to the problem of storing semantic sensor data:
\begin{itemize}
	\item The \textit{CSSD} approach using factorization techniques;% for a compact representation of semantic sensor data described using the SSN ontology;
	\item Tabular representations of semantic sensor data to scale-up to large datasets;
	\item SPARQL query rewriting techniques against factorized sensor data; and
	\item An empirical evaluation of the proposed compact representations demonstrating the effectiveness and efficiency of the \textit{CSSD} approach.
\end{itemize}  

The article is structured as follows:
We motivate the research problem in \autoref{sec:example}, and review existing work in \autoref{sec:related}. A formal description of our approach is discussed  in \autoref{sec:approach}, and tabular representations in \autoref{sec:tabular}.
We present the experimental study in \autoref{sec:eval} with an outlook on future work in \autoref{sec:conclusion}.

\section{Motivating Example}
\label{sec:example}
The \textit{MesoWest LinkedObservation}\footnote{\url{http://wiki.knoesis.org/index.php/LinkedSensorData}} datasets encompass sensor data containing weather observations during hurricane and blizzard seasons in the United States. 
Observations incorporate measurements of several weather phenomena, e.g., wind direction, snowfall, wind speed, rainfall, humidity, and temperature.
These weather observations from sensors are semantically described using the Semantic Sensor Network (SSN) ontology. These \textit{LinkedObservations} enclose almost two billion RDF triples semantically describing sensor data collected during major active storms in the United States since 2002.
The RDF sensor data describing the weather observations during the storm season in the year 2004 contains 108,644,568 RDF triples representing 11,648,607 observations about different weather phenomena, i.e., precipitation, rainfall, wind direction, temperature, and relative humidity. \autoref{fig:OriginalRDFGraph} illustrates the RDF graph of sensor data describing pressure, wind direction, rainfall, temperature, and visibility observations, as well as observation timestamps from the MesoWest dataset during the year 2004.
The RDF graph depicts 46,341 RDF triples semantically describing 5,149 sensor observations.
The RDF triples associated with the same measurement value are represented by the same color nodes and edges in the graph.
The RDF triples affiliated with timestamps are also represented with the same color nodes and edges.
The RDF graph statistics, shown in ~\autoref{fig:statsOriginal}, indicate the existence of remarkably redundant inter-connectivity among the RDF nodes.
The RDF graph and the statistics present that the RDF triples are replicated with the redundant measurement values.
Also, each sensor observation is related to seven neighbors in average, i.e., observations are semantically described using seven RDF triples in average.

\autoref{fig:NTTotal} depicts the number of RDF triples per distinct measurement value within the RDF dataset.
Rainfall measurement value, {\tt 6.55 cm}, is highly repeated and is related to 15,552 RDF triples, and {\tt $113^{\circ}$} wind direction is the second highly repeated value and is affiliated with 13,941 RDF triples.
Likewise, the number of RDF triples related with unique values can be noticed for other climate phenomena, e.g., temperature, pressure, and visibility, and corroborate the {\it natural intuition} that the number of {\it observations} is much higher than the the number of {\it distinct measurement values}. 
We exploit this natural intuition of semantic sensor data, and present a compact representation. RDF triples of repeated measurements values are {\it factorized} in these compact representations, and are added to the dataset only once without losing any information initially encoded in the sensor data.
Unlike other RDF data compression techniques, the semantics of observations are utilized to factorize the semantic sensor data. The factorized representations provide efficient storage over diverse RDF implementations, and queries can be directly executed over factorized RDF datasets. 
To scale up to large datasets, tabular representations, based on the factorization, can be used to exploit Big Data technologies for storage and management of large amount of semantic sensor data.

\section{Related Work}
\label{sec:related}
Semantic Web and Big Data communities have been working for better storage and processing of large datasets. RDF compression techniques~\cite{alvarez2011compressed,bok2019provenance,FernandezMGPA13,meier2008towards,pan2014ssp,pichler2010redundancy} are devised, as well as, Big Data tools are exploited in \cite{du2012hadooprdf,khadilkar2012jena,mami2016towards,nie2012efficient,papailiou2013h,punnoose2012rya,schatzle2013pigsparql} to efficiently process RDF data.  Furthermore, column-oriented stores \cite{idreos2012monetdb,macnicol2004sybase,stonebraker2005c,zukowski2006super} exploit fully decomposed storage model \cite{copeland1985decomposition} to scale-up to large datasets, and data factorization based query optimization techniques are proposed in \cite{BakibayevKOZ13}.
\subsection{RDF Data Compression}
%A large number of organizations and enterprises facilitate re-usability and integration of RDF data by numerous applications. Data integration by numerous applications leads to an accelerated growth in the amount of RDF data causing performance bottlenecks for RDF management systems. 
% A rule-based RDF data compression method, presented by Joshi et al. ~\cite{joshi2013logical}, generates decompression rules and removes the RDF triples that can be inferred using these rules. 
%RDF datasets are divided into two smaller disjoint datasets, i.e., an active and a dormant dataset. An active dataset consists of compressed triples to which decompression rules are applied for inferring the new triples. On the contrary, a dormant dataset comprises uncompressed triples to which no rules are applied.
A user specific approach to minimize RDF graphs by defining Datalog rules to remove the irrelevant RDF data is proposed by Meier~\cite{meier2008towards}. 
%The approach uses constraints to maintain data consistency before and after RDF graph minimization. Instead of generating new RDF triples, Datalog rules defined by the user are used to remove RDF data from RDF graphs that is not required for the application.  These datalog rules are used to reconstruct the deleted data. 
Similarly, Pichler et al.~\cite{pichler2010redundancy} study the complexity of RDF minimization in presence of constraints, rules, and queries. These approaches require data decompression to process and manage the compressed RDF datasets.
%An approach to identify and reduce semantic, syntactic and symbolic redundancies is proposed by Pan et al.~\cite{pan2014ssp}. 
A graph pattern based logical compression technique, proposed by Pan et al.~\cite{pan2014ssp}, replaces bigger graph patterns by smaller graph patterns and generates a sequence of bits for each graph pattern. 
%The generated bit sequence is composed of two components, i.e., the graph pattern itself shared by the instances and the sequence of instances of the graph pattern. 
Similarly, Fernandez et al.~\cite{FernandezMGPA13} compresses and describes RDF data in binary format in terms of header, dictionary, and triples.
The header contains compression relevant metadata, the dictionary contains identifiers of data values and triples represent the collection of data identifiers. 
A compressed RDF structure, k\textsuperscript{2}-triples, presented by \'Alvarez-Garc\'ia et al. \cite{alvarez2011compressed},
vertically partitions RDF triples, and utilizes k\textsuperscript{2}-trees~\cite{brisaboa2009k} to create indexes for each partition. Bok et al.~\cite{bok2019provenance} present RDF provenance compression technique by exploiting dictionary encoding. 
%Frequently repeated subgraphs are extracted from the numeric version of RDF triples and are stored as reference patterns.
These approaches provide effective solutions for RDF data compression. 
However, customized engines are required to execute queries over the compressed RDF data, and data management tasks demand decompression techniques to be performed over the compressed data.
Contrary, we propose factorization techniques that generate a compact representation by exploiting properties of semantic sensor data, where queries can be executed directly over the compact representations.
Since factorization and compression techniques are independent, and do not directly intervene with each other, both can be exploited in conjunction.

\subsection{Big Data Tools and RDF}
%With the tremendous growth of semantic data, the problem of storing and processing large-scale semantic data has become of paramount importance. Semantic Web researchers are exploiting big data frameworks to efficiently store and process the continuously growing RDF data~\cite{du2012hadooprdf,khadilkar2012jena,mami2016towards,nie2012efficient,papailiou2013h,punnoose2012rya,schatzle2012cascading,schatzle2013pigsparql}. These approaches offer RDF storage and indexing schemes for efficient RDF data processing over Big Data frameworks.
Relational representations of RDF data over big data storage technologies, i.e., Parquet and MongoDB, are presented by Mami et al.~\cite{mami2016towards}, where a table for each RDF class is created, representing class properties as attributes. %Additionally, a table for each class is created that contains a type column of each attribute. 
Du et al.~\cite{du2012hadooprdf} combine Hadoop framework and an RDF triple store, Sesame, to achieve scalable RDF data analysis. RDF data is partitioned in such a way that all the triples with the same predicate are allocated to the same partition. Jena-HBase~\cite{khadilkar2012jena} provides a variety of RDF data storage layouts for HBase and all operations over RDF graph are converted into the underlying layout operations. 
%The RDF storage layouts include the simple (three tables each with an index over subject, predicate, and object), vertical partitioned, indexed, vertical partitioned and indexed, hybrid, and hash layouts. 
%Sch{\"a}tzle et al.~\cite{schatzle2014sempala} utilize Parquet columnar storage format for RDF data storage over Hadoop. %, and SPARQL queries are compiled as SQL queries.
%to be executed using Impala. RDF data is stored in a single table containing all properties as the attributes of the table. 
Sch{\"a}tzle et al.~\cite{schatzle2013pigsparql} present PigSPARQL, a SPARQL query processing framework using Hadoop MapReduce over large RDF graphs. 
A scalable RDF data management system developed by Punnoose et al.~\cite{punnoose2012rya} presents storage methods and indexing, % and query processing techniques using conventional query languages, i.e., SPARQL. 
where RDF data is stored as a pair of a key and a corresponding value, and %The key consists of RowID, column and Timestamp. 
RDF triples are indexed using SPO, POS, and OSP. %The storage is compact because RDF triples are stored only on RowID and no data is stored in column, timestamp, and value fields. 
Papailiou et al.~\cite{papailiou2013h} present an RDF store to efficiently perform distributed Merge and Sort-Merge joins using multiple indexing over HBase, where indexes that are compressed using dictionary encoding. 
 Nie et al.~\cite{nie2012efficient} study the efficient RDF partitioning and indexing schemes to process RDF data in distributed way using MapReduce. %Using horizontal partitioning RDF data is partitioned in such a way that all the triples with the common hash value of subject are contained in the same file. Vertical partitioning combines all the triples with the same predicate in the same file. Clustered property partitioning organizes RDF triples with the same subjects and then partitions RDF triples into clusters using the same property sets. 
 %Sch{\"a}tzle et al.~\cite{schatzle2012cascading} present RDF storage schema for HBase to efficiently process joins using MapReduce. %RDF data is stored in two tables; one with subject as row key and the other with object as the row key. %HBase filter API is used to filter the data on server side to avoid unnecessary data transfer.  
 We propose factorization techniques for the RDF sensor data where the RDF triples related to the redundant values are factorized. The proposed tabular-based representations of factorized RDF graphs, i.e., factorized tables and CT based tables, scale up to large datasets by leveraging the column-oriented Parquet storage format. The tabular representation of the factorized RDF graphs remove data redundancies and improve the storage and query processing using Big Data tools.
\subsection{Relational Data Compression}   
%Performance of database systems is related to the efficiency of storing the data on primary storage and the improved IO bandwidth. For this reason, database community has explored several alternatives for storage implementations. 
Column-stores have gained attention for being able to efficiently store data and improve the IO bandwidth for large-scale data intensive applications. Early efforts include C-Store~\cite{stonebraker2005c}, SybaseIQ~\cite{macnicol2004sybase}, MonetDB~\cite{idreos2012monetdb}, %VectorWise~\cite{boncz2005monetdb},  
and lightweight data compression by Zukowski~\cite{zukowski2006super}. %Column-oriented systems completely vertically partition a relation into a collection of individual columns around each attribute. By storing each column separately, the only columns required to answer a particular query are fetched from the disk rather than the entire row. Similarly, IO and memory bandwidth is improved by transferring only the required column data. Furthermore, column-oriented storage applies compression techniques over the data in a column related to the same type more naturally. 
Various compression techniques are exploited in C-Store~\cite{stonebraker2005c} to support several column sort-orders without space explosion. %Multiple sort orders open opportunities for optimization. 
C-Store compresses each column using one of the four encoding schemes defined based on the order of values in the column. %, i.e., column values are ordered either by the values in the column or by the corresponding values in some other column. 
%The encoding schemes include the self-order or foreign-order of the column values containing few or many distinct values. %For each encoded representation B-tree data structures are used for indexing. 
Similarly, SybaseIQ~\cite{macnicol2004sybase} uses column-store to perform complex analytics efficiently on massive amounts of data, and optimizes workloads across multiple servers through multi-node, shared storage, and parallel database system. 
MonetDB~\cite{idreos2012monetdb} exploits column-store technology to efficiently perform analytics over the large collections of data. %VectorWise~\cite{boncz2005monetdb} provides a query engine for MonetDB that uses vectorized execution to achieve high CPU efficiency. % and to scale-up towards non main-memory datasets.
 Furthermore, light-weight compression is used to keep intermediate results in memory for reuse.
Zukowski~\cite{zukowski2006super} proposes lightweight data compression techniques over the column-stores in order to speedup the data-intensive query processing. %  The proposed techniques implement compression and decompression algorithms that lack the if-then-else constructs, and absence of dependencies between the values help to make them loop-pipelinable. Additionally, data distribution with outliers is handled for better compression ratio, and compression and decompression are used at the CPU and RAM storage levels. 
These compression techniques exploit column-oriented stores that use fully decomposed storage model by Copeland et al.~\cite{copeland1985decomposition}, where \emph{n-array} relations are decomposed into \emph{n} binary relations, i.e., a pair of attribute value and an identifier. %Every binary relation is stored as a pair consisting of an attribute value and a system generated identifier that corresponds to the relational tuple containing the attribute value. Similarly, 
Two copies of each binary relation are stored increasing the storage requirements. 
%Moreover, for each attribute the identifier is stored twice, increasing the storage requirements by a factor of two.
 Our approach generates a factorized RDF graph where data redundancies are reduced. %These factorized RDF graph representations replace repeated properties and corresponding objects with compact molecules. 
 Hence, factorized representations reduce the storage requirements for the decomposition storage model.

\subsection{Data Compression based Query Optimization}
Factorization techniques have been utilized for optimization of relational data and SQL query processing~\cite{BakibayevKOZ13,BakibayevOZ12}. 
Existing approaches proposed compact representations of relational data, obtained by applying logical axioms of relational algebra, e.g., distributivity of product over union, and commutativity of product and union. 
Bakibayev et al.~\cite{BakibayevOZ12} present an in-memory query engine to run select-project-join queries over factorized data. The query results are expressed using factorized representations in terms of singleton relations, product, and union. The compact representations are obtained by algebraic factorization using distributivity of product over union, and the commutativity of product and union. These factorized representations form a nested structure containing the attributes from the schema, and are referred as factorization tree. A set of operators for selection and projection are proposed that map the factorized representations and generate efficient query plans.
Similarly, Bakibayev et al.~\cite{BakibayevKOZ13} improve the performance of relational processing for aggregate and ordering queries using the distributivity of product over union to factorize relations as in the factorization of logic functions~\cite{brayton1987factoring}. Factorized representations reduce the number of computations required for the evaluation of aggregation functions, i.e., sum, count, avg, min, max, likewise evaluation of aggregation functions as sequences of partial aggregations over the factorized representation speedup the processing. To evaluate order-by-queries, factorized tables are restructured with a constant delay enumeration.
Therefore, queries can be executed in factorized relational data, and efficient execution plans can be found to speed up execution time. 
We build on these experimental results and propose factorization technique tailored for semantically described sensor data. 
The \textit{CSSD} approach exploits the semantics encoded in RDF sensor data to compactly represent RDF triples, reduce redundancy, and facilitate query execution. 

\begin{figure*}[ht!]
 \centering
    \subfloat[Semantic Sensor Network (SSN) Ontology]{ \includegraphics[width=0.6\linewidth]{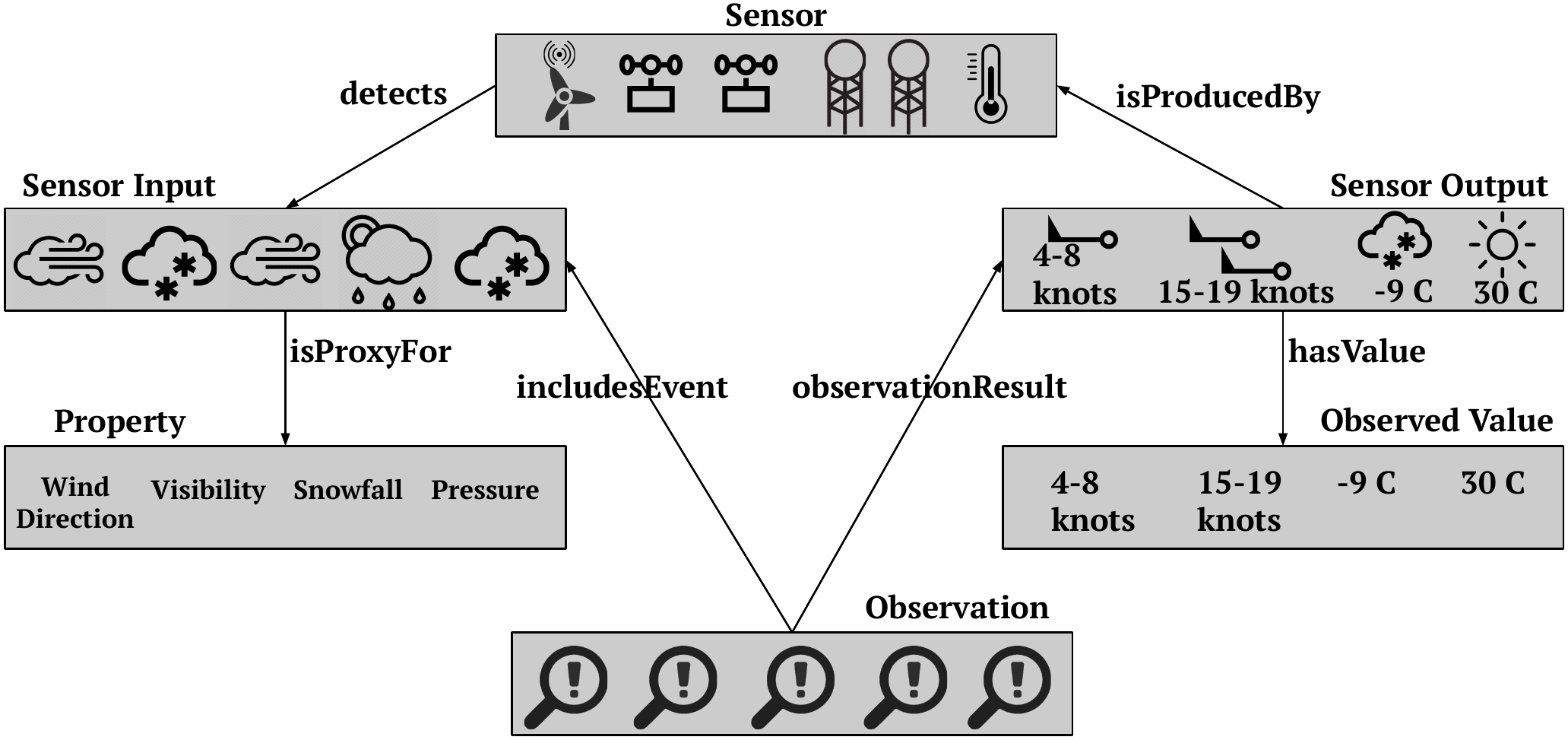}\label{fig:ssnOnt}}
  \hspace{5pt}\subfloat[RDF Molecules]{\includegraphics[width=0.35\linewidth]{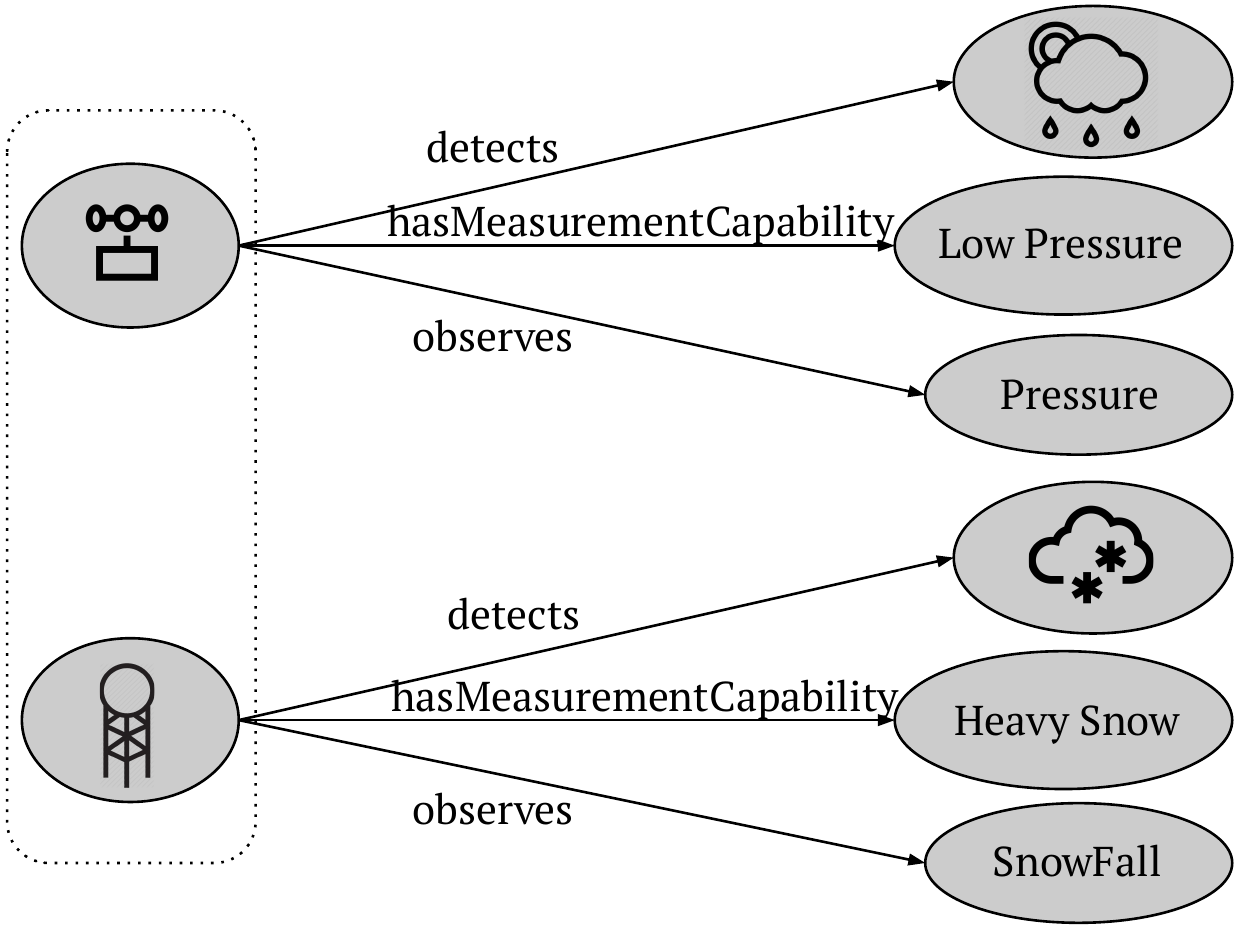}\label{fig:rdfMol}}
  \caption{{\bf Overview of the Semantic Sensor Network (SSN) Ontology}.  (a) The SSN Ontology is composed of 50 classes and 55 properties to describe sensor observations; a portion of the SSN classes and properties is presented. (b) An RDF graph with two subject molecules in the class Sensor; for clarity URIs are omitted.}
  \label{fig:ssnOntology}
\end{figure*}

\section{The Semantic Sensor Data Factorization Approach}
\label{sec:approach}
\subsection{Preliminaries}
The Semantic Sensor Network (SSN) Ontology~\cite{compton2012ssn}, developed by the W3C Semantic Sensor Network Incubator Group\footnote{https://www.w3.org/2005/Incubator/ssn/}, is an OWL ontology that consists of 50 RDF classes and 55 properties to semantically describe sensor data in terms of observations, observed properties, features of interest, and measurement units and observed values.
% allows for the description of sensor devices, their capabilities, observations, and other sensor-related concepts. %The SSN ontology is commonly exploited to tackle interoperability and data semantic enrichment problems in IoT applications~\cite{ali2015citybench,gao2014semantic,henson2009ontological,PhuocQQNH16} when integrating heterogeneous data sources.
%The SSN ontology consists of 50 RDF classes and 55 properties to semantically describe sensor data in terms of observations, observed properties, features of interest, and measurement units and observed values.
A portion of the SSN ontology, illustrated in ~\autoref{fig:ssnOnt}. %, comprises classes and properties to describe observations and measurements observed by sensors. 
Sensors generate observations by detecting certain properties of features of interest and produce the observed values as sensor output.
Given disjoint infinite sets \textbf{I},  \textbf{L}, \textbf{B} of IRIs, literals, and blank nodes, respectively, a tuple $(s\;p\;o) \in
(I \cup B)\;$ $\times \;I\; \times\; (I \cup B \cup L)$ is called an RDF triple. An RDF graph $G=(V_G,$ $E_G,L_G)$ comprises RDF triples, where  $V_G$  is a set of nodes in \textbf{I} $\cup$  \textbf{B} $\cup$ \textbf{L}, $E_G$ is a set of edges representing RDF triples, and $L_G$ is a set of edge labels in \textbf{I}~\cite{arenas2009foundations}. \autoref{fig:multiplicityExp} illustrates an RDF graph representing a portion of the RDF dataset from the weather observations.
Nodes correspond to resources representing sensor observations, timestamps, measurements, and literals. %Furthermore, literals are also illustrated as nodes in the RDF graph and properties typically originate from a variety of RDF vocabularies that involve the Semantic Sensor Network (SSN) Ontology. 
Edges in RDF graphs represent RDF triples and connect the nodes in RDF graphs using properties from the SSN ontology.
We ignore name of the properties, prefixes, and replace long URLs by short identifiers for clarity.
We refer to such an RDF graph described using the SSN ontology in this paper as an SSN RDF graph.
%\hl{Given disjoint infinite sets \textbf{I}, \textbf{B}, \textbf{L} of IRIs, blank nodes, and literals, respectively, an RDF graph is a pair $G=(V_G,E_G)$, where  $V_G$  is a set of nodes in \textbf{I} $\cup$  \textbf{B} $\cup$ \textbf{L}, and $E_G$ is a set of RDF triples} \cite{arenas2009foundations}. Figure\autoref{fig:multiplicityExp} \hl{presents an RDF graph that corresponds to a portion of the RDF dataset from the storm season in year 2004. Nodes correspond to resources representing observations, measurements, and timestamps. Further, literals are also represented as nodes in the RDF graph and properties typically stem from a variety of RDF vocabularies that include the Semantic Sensor Network (SSN) Ontology.} Edges in RDF graphs represent RDF triples and connect the nodes in RDF graphs using properties from the SSN ontology. We ignore name of the properties, prefixes, and replace long URLs by short identifiers for clarity. The RDF graph that includes properties from the SSN ontology and the instances correspond to the classes of the SSN ontology, we refer to such an RDF graph in this paper as an SSN RDF graph.
RDF graphs are usually composed of entity description sub-graphs, sometimes also referred to as Concise Bounded Descriptions (CBD)\footnote{https://www.w3.org/Submission/CBD/}. 
These subgraphs are named {\it RDF subject molecules} defined as follows: 
Given an RDF graph G, a subgraph $M$ of $G$ is an \textit{RDF molecule}~\cite{FernandezLC14} iff the RDF triples of $M=\{t_1,\dots, t_n\}$ share the same subject, i.e., $\forall$ $i,j \in \{1,..,n\}$ ($\mathit{subject}(t_i) = \mathit{subject}(t_j)$).
Figure~\ref{fig:rdfMol} presents an RDF graph with two RDF subject molecules. 
Each RDF molecule consists of three RDF triples connected to the same subject, which represents an instance of the sensor class.
%Each instance of the sensor class is described by the RDF triples in terms of input, observed property, and measurement capability.
%For simplicity, we omit the long URIs in the figure, and will refer to RDF subject molecules as {\it molecules} in the rest of the paper.
We will refer to RDF subject molecules as {\it molecules} in the rest of the paper.

\begin{figure*}[ht!]
\centering
     \subfloat[RDF Graph $G$]{
      \includegraphics[width=.5\linewidth]{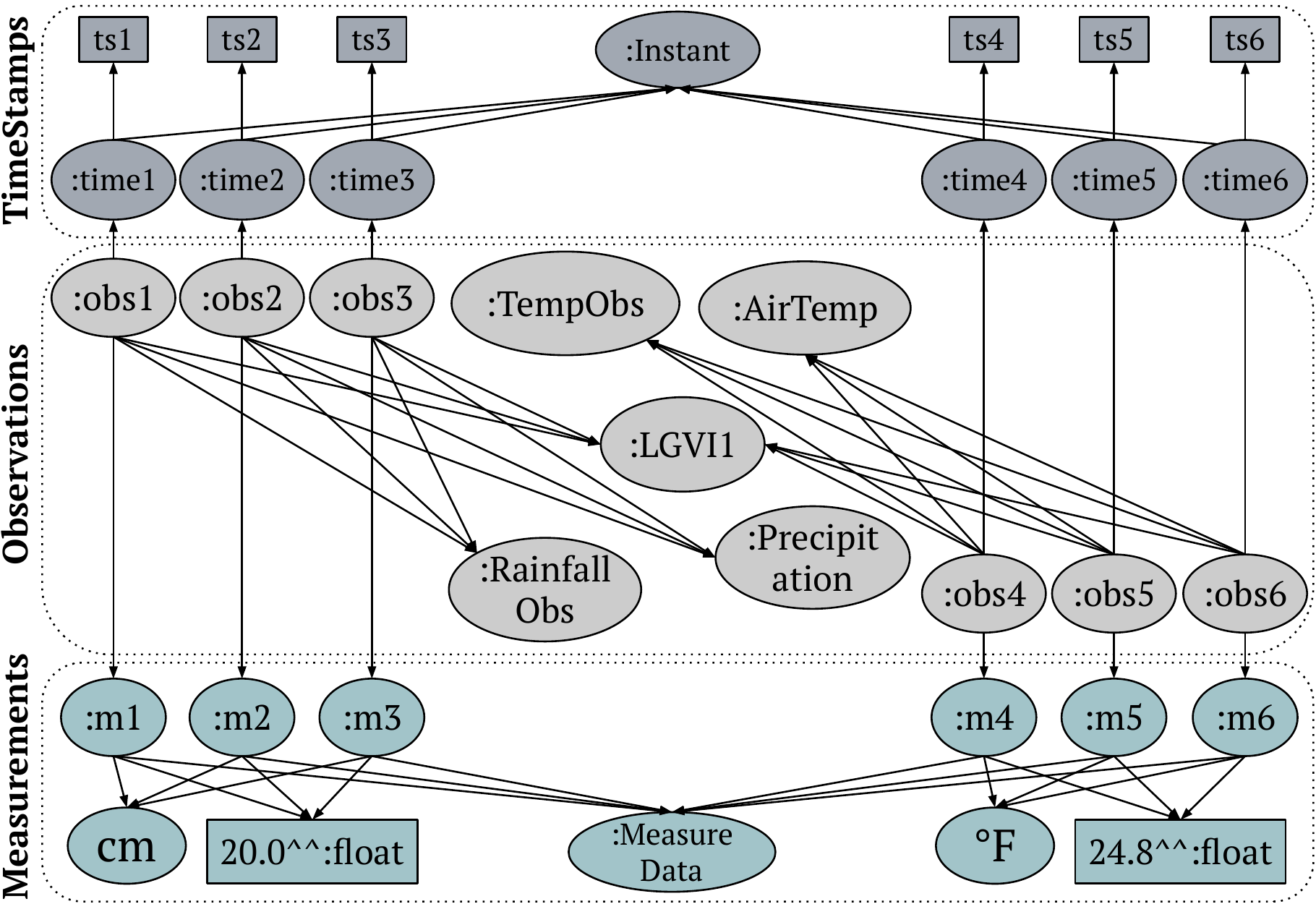}
      \label{fig:multiplicityExp}}
    \subfloat[Obs. Molecules]{
      \includegraphics[width=.2\linewidth]{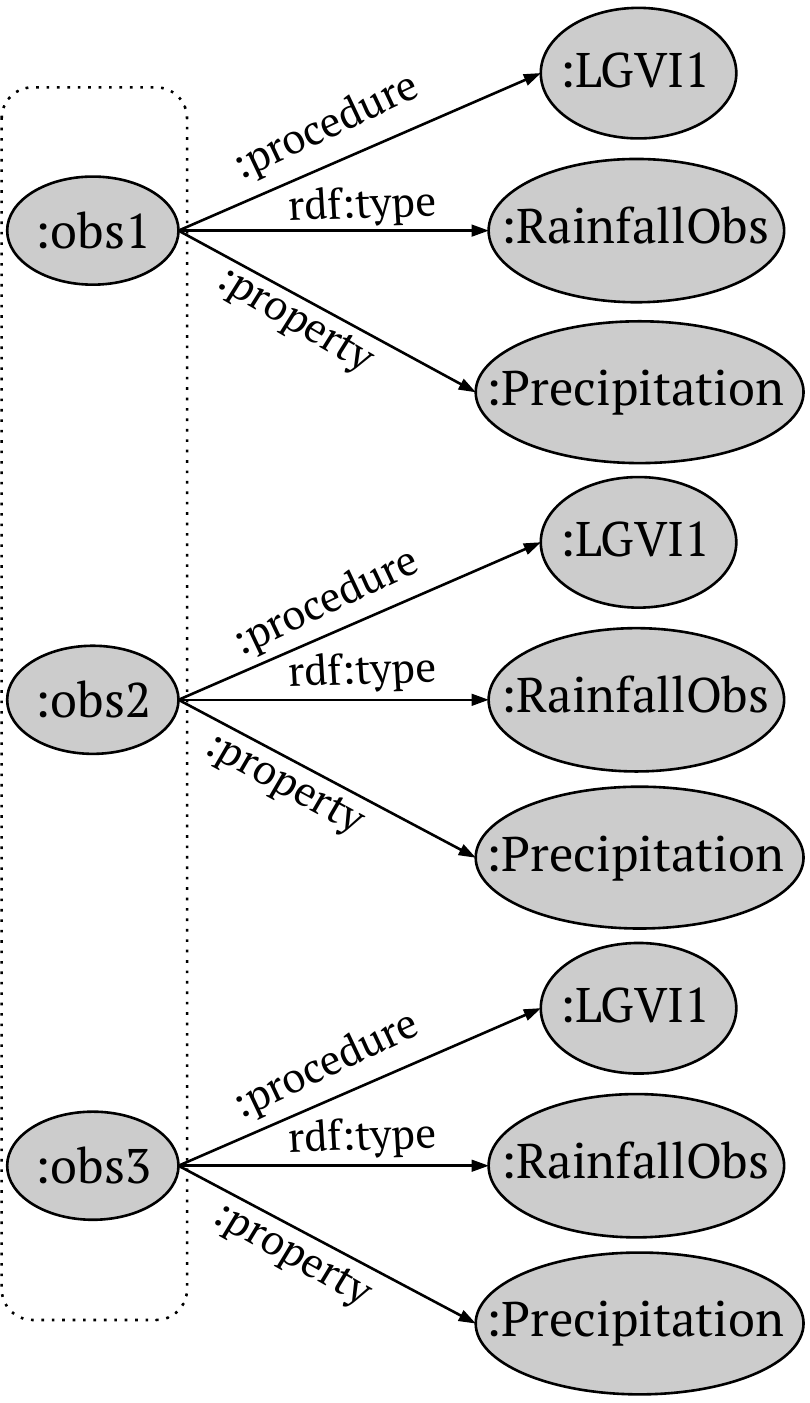}
      \label{fig:obsMolecule}}
    \subfloat[Meas. Molecules]{
      \includegraphics[width=.2
      \linewidth]{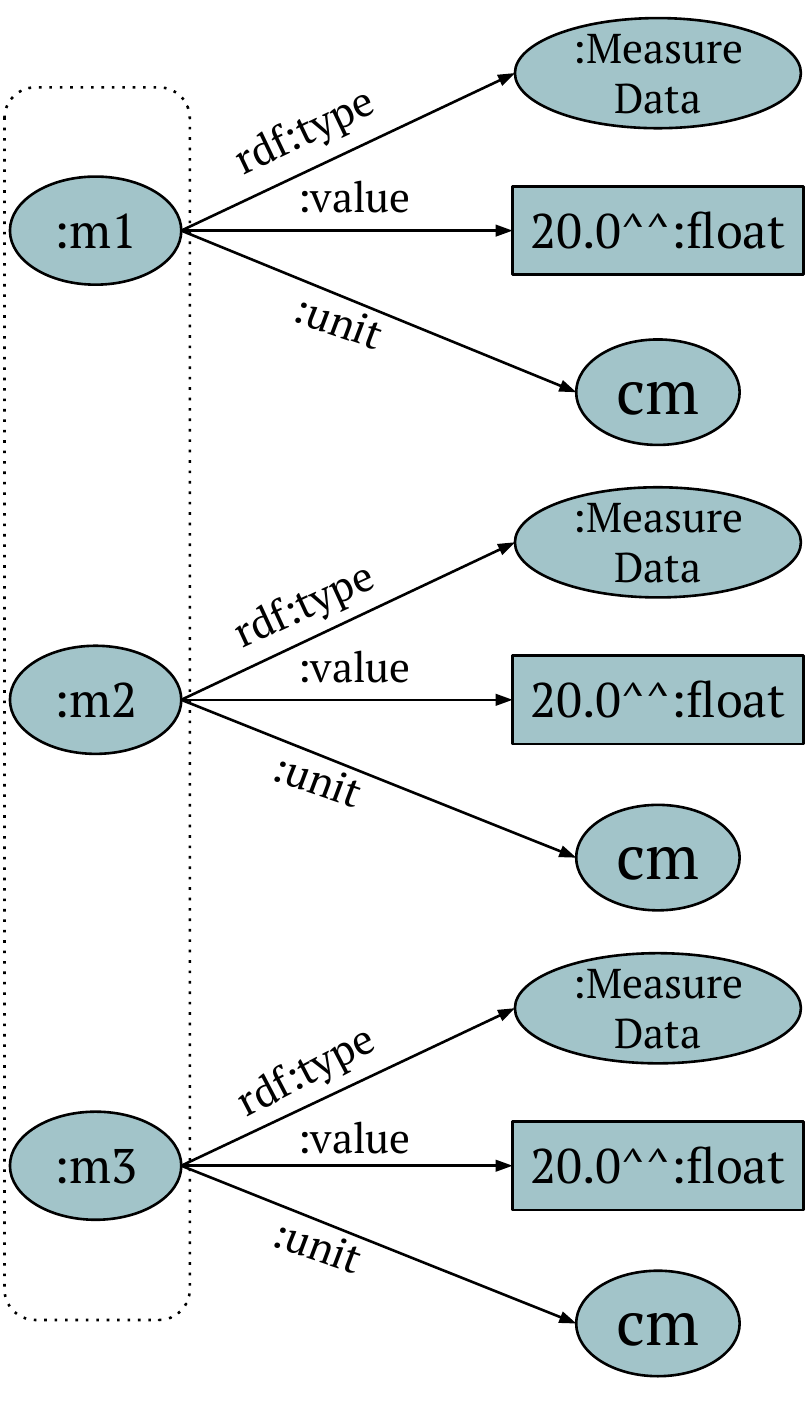}
      \label{fig:measMolecule}}
    \caption{{\bf Example of a Simplified RDF.} Several RDF triples are related to the same measurement values, for simplicity URIs are not presented. (a) RDF graph has $M_m(v,u|G)=3$  and  $M_o(s,p,pp,v,u|G)=3$ for the values $20.0cm$ and $24.8^{\circ}F$. (b) Three observation (Obs.) molecules; (c) Three measurement (Meas.) molecules.}
\end{figure*}

\subsection{Problem Statement} 
The concept of RDF molecule is utilized to devise observation and measurement molecules based on the SSN Ontology. Moreover, we present the concept of multiplicity.
Building on these definitions the problem tackled in this work is defined.  

\begin{definition}[Observation Molecule]
An observation molecule $OM$ is a set of RDF triples that share the same subject of type observation class, i.e., $OM$= $(obs \; \texttt{rdf:type} \; \textit{:Observation})$, $(obs \; \texttt{:procedure} \; proc)$, $(obs \; \texttt{:property} \; pp)$.
\end{definition}
 Figure~\ref{fig:obsMolecule} presents three observation molecules, each consists of three RDF triples describing an observation subject. Each observation subject is described in terms of observation type, observed property and the observation procedure.
 
\begin{definition}[Measurement Molecule]
A measurement molecule $MM$ is a set of RDF triples that share the same measurement subject, i.e., $MM$= $(m \; \texttt{rdf:type} \;\\ \textit{:MeasureData})$, $(m \; \texttt{:value} \; val)$, $(m \; \texttt{:unit} \; uom)$.
\end{definition}

Figure~\ref{fig:measMolecule} presents three measurement molecules, each consists of three RDF triples having the same measurement subject. Each measurement is described in terms of measured value and unit.
Class Templates are the abstract descriptions of the triples in RDF graphs and are defined as follows:

\begin{definition}[Class Template (CT)] Given a class $C$ in an RDF graph $G$, a Class Template is a $4-tuple=<C,SP,IntraL,InterL>$, where, $SP$ is a set of properties in $C$, $IntraL$ is a set of pairs ($p,C_j$) such that $p$ is an object property with domain $C$ and range $C_j$ in the same dataset, and $InterL$ is a set of pairs ($p,C_k$) such that $p$ is an object property with domain $C$ and range $C_k$ in different datasets.
A Class Template is a simplification of an RDF Molecule Template~\cite{endris2017mulder}.
\end{definition}

Figure~\ref{fig:rdfmtLinkExam} shows class templates (CT) extracted from Figure~\ref{fig:multiplicityExp} around the \texttt{:TempObs}, \texttt{:RainfallObs}, \texttt{:Instant}, and \texttt{:MeasureData} classes (Figure~\ref{fig:rdfmtExample}). Moreover, Figure~\ref{fig:rdfmtLinking} shows intra-link between \texttt{:MeasureData} and \texttt{:TempObs} using \texttt{:result}, and inter-link of \texttt{:TempObs} and \texttt{:RainfallObs} to \texttt{:Instant} using \texttt{:samplingTime}. 

\begin{figure*}[ht!]
 \centering
    \subfloat[Class Templates (CTs)]{ \includegraphics[width=0.45\linewidth]{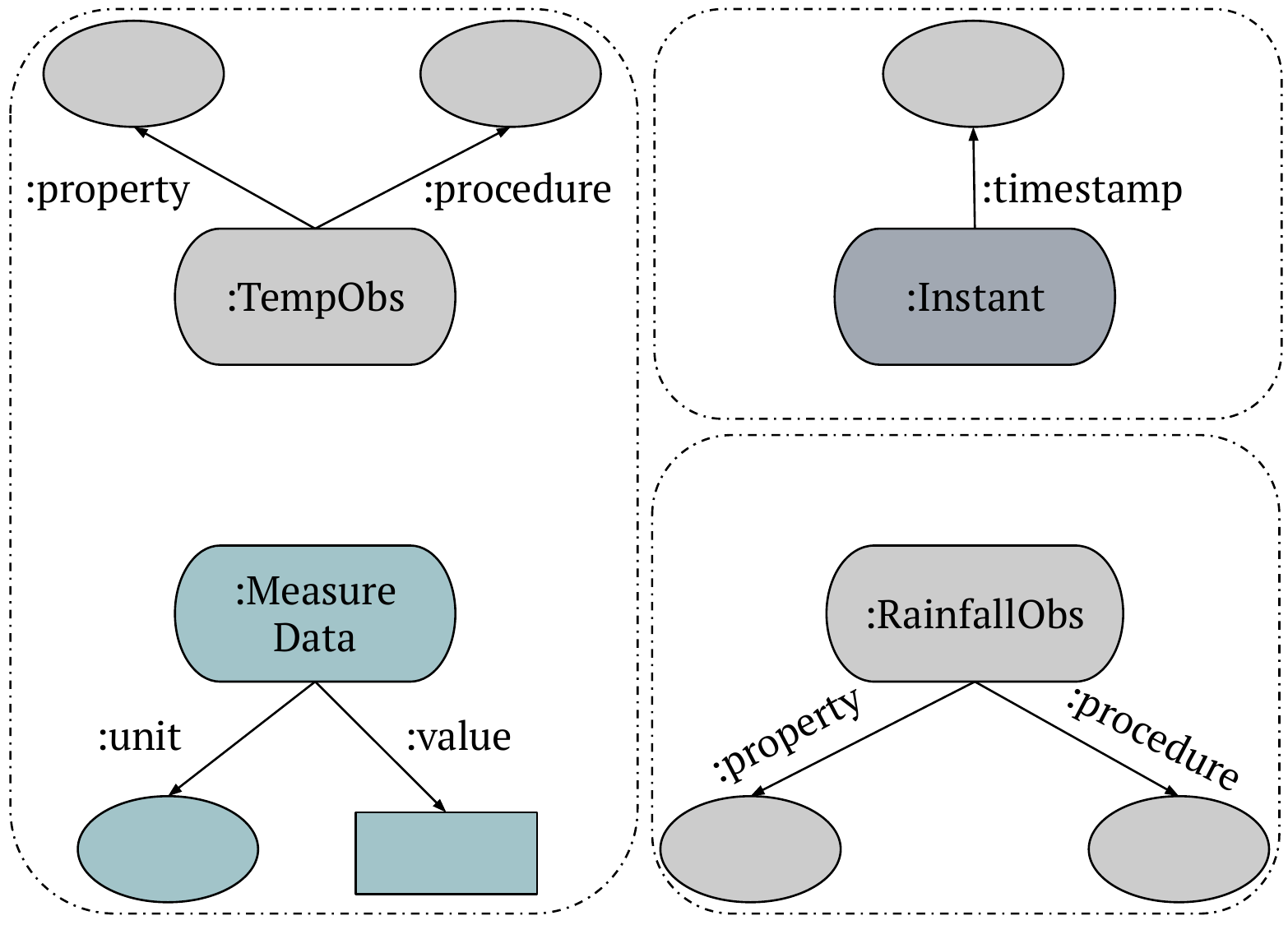}\label{fig:rdfmtExample}}
  \hspace{5pt}\subfloat[CT Intra- and Inter-dataset Linking]{\includegraphics[width=0.52\linewidth]{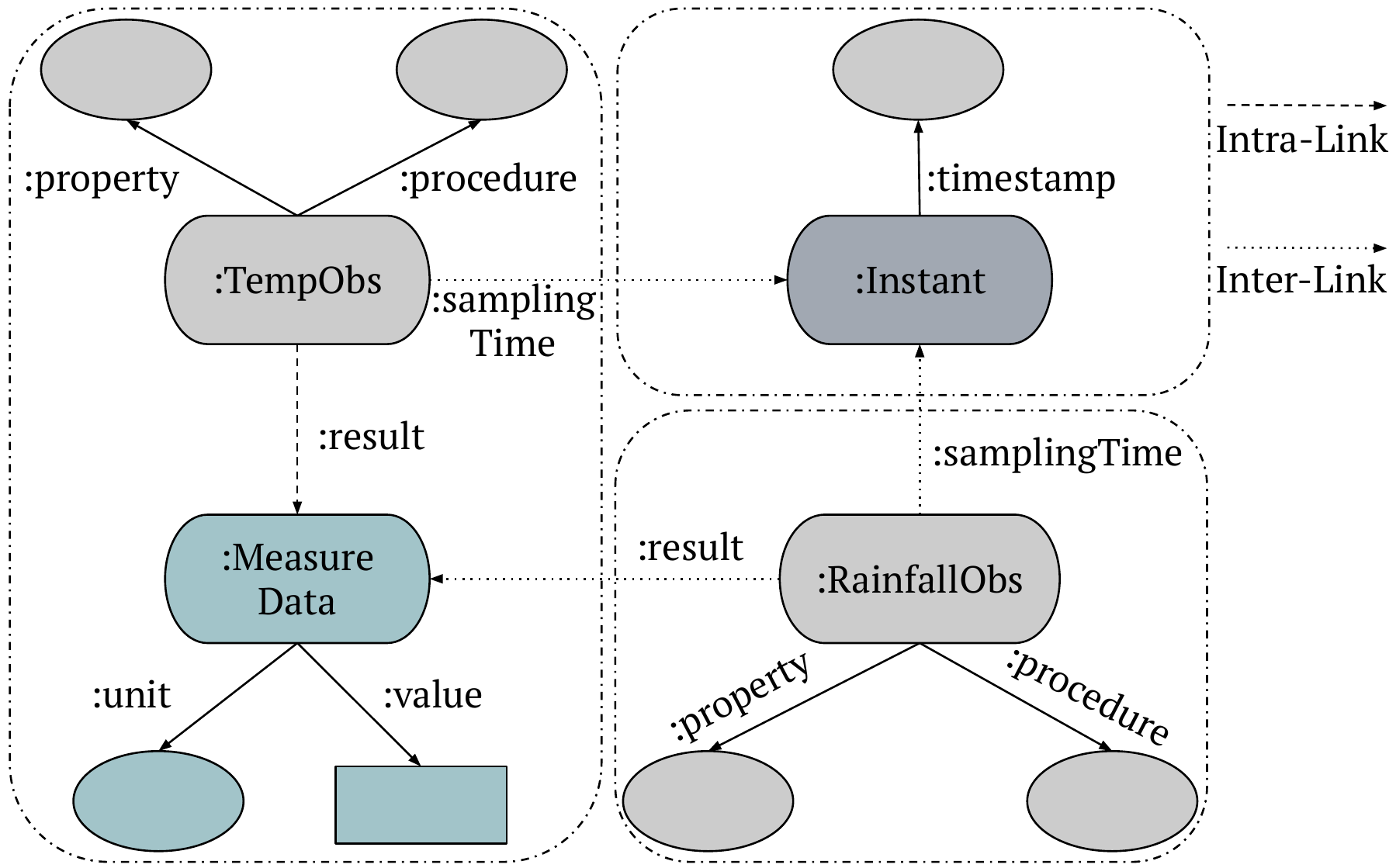}\label{fig:rdfmtLinking}}
  \caption{{\bf CTs and CT Linking}.  (a) Four Class Templates are extracted from the RDF graph in Figure~\ref{fig:multiplicityExp} around the classes \texttt{:TempObs}, \texttt{:RainfallObs}, \texttt{:MeasureData}, and \texttt{:Instant}; (b) inter- and intra-links between class templates. %\texttt{:TempObs} is intra-linked to \texttt{MeasureData} by property \texttt{:result}. \texttt{:TempObs} and \texttt{:RainfallObs} are  inter-linked to \texttt{:Instant} using \texttt{:samplingTime}, and \texttt{:RainfallObs} is inter-linked to \texttt{:MeasureData} using \texttt{:result}.
  }
  \label{fig:rdfmtLinkExam}
\end{figure*}

\begin{definition}[Measurement Multiplicity~\cite{karim2017large}] Given an RDF graph $G$ of sensor data using the SSN ontology.
Given a resource \texttt{uom} corresponding to a measurement unit, and a literal value \texttt{val}, the measurement multiplicity of \texttt{uom} and \texttt{val} in $G$, $M_m(val,uom|G)$, is defined as the number of measurements have same value $val$ and measurement unit $uom$ in $G$.
\[\arraycolsep=1.2pt
\begin{array}{ll}
M_m(val,uom|G)= |\{m|&(m \;\texttt{rdf:type} \;\texttt{:MeasureData}) \\& \in 
G, (m\; \texttt{:unit}\; uom) \in G,\\
&(m\; \texttt{:value}\; val) \in G \}|
\end{array}
\]
\end{definition}

In the RDF graph in \autoref{fig:multiplicityExp}, three measurements, i.e.,  {\tt \footnotesize :m1}, {\tt \footnotesize :m2}, and {\tt \footnotesize :m3}, are related to the unit {\tt \footnotesize cm} and value {\tt \footnotesize 20.0}. Therefore, the measurement multiplicity of {\tt \footnotesize cm} and {\tt \footnotesize 20.0} is 3. Similarly, the measurement multiplicity of {\tt \footnotesize $^{\circ}$F } and {\tt \footnotesize 24.8} is 3.

\begin{definition}[Observation Multiplicity~\cite{karim2017large}] Given an RDF graph $G$ of sensor data described using the SSN ontology.
Given resources \texttt{proc}, \texttt{ph}, \texttt{pp}, and \texttt{uom} corresponding to a procedure, an observed phenomenon, observed property, and measurement unit, and a literal value \texttt{val}, the multiplicity of an observation \texttt{obs} for \texttt{uom} and \texttt{val} in $G$, $M_o(proc,ph,pp,val,uom|G)$, is defined as the number of observations about the property \texttt{pp} of the observed phenomenon \texttt{ph}, sensed by \texttt{proc}, that have the same value \texttt{val} and unit of measurement \texttt{uom} in $G$.
\[\arraycolsep=3.8pt
\begin{array}{ll}
M_o(proc,ph,pp,val,uom|G)&=|\{obs |\\&(obs\; \texttt{rdf:type}\; ph) \\&\in G, (obs \;\texttt{:procedure}  \; \\&proc) \in G, (obs\; \texttt{:prop}\\&\texttt{erty}\; pp) \in G, (obs \; \\&\texttt{:result}  \; m) \in G, (m \\&\;\texttt{rdf:type} \;\texttt{:Measure}\\&\texttt{Data})\in G,(m\; \texttt{:unit}\; \\&uom) \in G, (m\; \texttt{:value}\; \\&val) \in G \}|
\end{array}
\]
\end{definition}

In \autoref{fig:multiplicityExp}, the procedure {\tt \footnotesize :LGVI1}, the phenomenon {\tt \footnotesize :RainfallObs}, the property {\tt \footnotesize :Precipitation}, the measurement unit {\tt \footnotesize cm}, and the value {\tt \footnotesize 20.0} are associated with {\tt \footnotesize :obs1}, {\tt \footnotesize :obs2}, and {\tt \footnotesize :obs3}, and the observation multiplicity is 3. Similarly, the observation multiplicity  for {\tt \footnotesize :LGVI1}, {\tt \footnotesize :TempObs}, {\tt \footnotesize :AirTemp}, {\tt \footnotesize $^{\circ}$F}, and {\tt \footnotesize 24.8} is 3.

\begin{definition}[Compact Observation Molecule] Given a surrogate observation $oM$, a compact observation molecule $COM$ is a set of RDF triples that share the same surrogate observation $oM$, i.e., $COM$= $(oM \; \texttt{rdf:type} \;\textit{:Observation})$,$(oM \;\texttt{:procedure} \; proc)$, $(oM \; \texttt{:property} \; pp)$, such that the observation multiplicity of the procedure \texttt{proc} and the observed property \texttt{pp} is one.
\end{definition}

Figure~\ref{fig:compMol} presents a compact observation molecule, with a surrogate observation \texttt{:obsM1} that corresponds to the observations \texttt{:obs1}, \texttt{:obs2}, and \texttt{:obs3} in Figure~\ref{fig:obsMolecule}, and is associated with the observation multiplicity value one. %The compact observation molecule represents all the properties and related objects as in the observation molecules in Figure~\ref{fig:obsMolecule}. However, the redundant edges, repeatedly connecting the similar type of observations \texttt{:obs1}, \texttt{:obs2}, and \texttt{:obs3}  to the object using the same properties, are transformed into the edges connecting these properties and corresponding objects to the surrogate observation \texttt{:obsM1}.

\begin{definition}[Compact Measurement Molecule]
Given a surrogate measurement $mM$, a compact measurement molecule $CMM$ is a set of RDF triples that share the same surrogate measurement $mM$, i.e., $CMM$= $(mM \; \texttt{rdf:type} \; \textit{:MeasureData})$, $(mM \; \texttt{:value} \; val)$, $(mM \; \texttt{:unit} \; uom)$, such that the multiplicity of the value \texttt{val} and the unit \texttt{uom} is one.
\end{definition}

A compact measurement molecule for the measurement molecules in Figure~\ref{fig:measMolecule} is presented in Figure~\ref{fig:compMol} using a surrogate measurement \texttt{:mM1}. \texttt{:mM1} corresponds to \texttt{:m1}, \texttt{:m2}, and \texttt{:m3} in Figure~\ref{fig:measMolecule} and is associated with the multiplicity value one. 

\begin{definition}[A Factorized RDF Graph] Given an RDF graph $G=(V_G,E_G,L_G)$ representing sensor data described using the SSN ontology, a factorized RDF graph $G'=(V_{G'},E_{G'},L_{G'})$ of $G$ is an RDF graph where the following hold:
\begin{itemize}
\item Entities in $G$ are preserved in $G'$, i.e., $V_G \subseteq V_{G'}$. 

\item For each entity $obs$ in $V_G$ that corresponds to an entity of \texttt{:Observation} class over the class properties \texttt{:procedure} and \texttt{:property} and objects $proc$ and $pp$, respectively, in $G$, there is an entity $oM$ in $V_{G'}$ that corresponds to the surrogate observation of the compact observation molecule over the properties \texttt{:procedure} and \texttt{:property} and objects $proc$ and $pp$, respectively, in $G'$. 
\item For each entity $m$ in $V_G$ that corresponds to an entity of class \texttt{:MeasureData} over the properties \texttt{:value} and \texttt{:unit} and objects $val$ and $uom$, respectively, in $G$, there is an entity $mM$ in $V_{G'}$ that corresponds to the surrogate measurement of the compact measurement molecule over the properties \texttt{:value} and \texttt{:unit} and objects $val$ and $uom$, respectively, in $G'$.
\item There is a partial mapping $\mu_N$: $V_G \rightarrow V_{G'}$:
\begin{itemize}
\item Observation entities in $G$ are mapped to the surrogate observations in $G'$, i.e., $\mu_N(obs)$=$oM$.
\item Measurement entities in $G$ are mapped to the surrogate measurements in $G'$, i.e., $\mu_N(m)$=$mM$.    
\item The mapping  $\mu_N$ is not defined for the rest of the entities that are not instances of the \texttt{:Observation} or \texttt{:MeasureData} class in $G$.
\end{itemize}

\item For each RDF triple $t$ in ($s \; p \; o$) in $E_G$:
\mbox{}
\begin{itemize}
\item If $\mu_N(s)$ is defined and \texttt{:Observation} is the type of $s$, then the RDF triples  $(s \; \texttt{:observationOf} \; \mu_N(s))$ and $(\mu_N(s)\; \texttt{rdf:type} \; \texttt{:Observation})$ is in $E_{G'}$.
\item If $\mu_N(s)$ is defined and \texttt{:MeasureData} is the type of $s$, then the RDF triple $(\mu_N(s)\; \texttt{rdf:type} \; \texttt{:MeasureData})$ belong to $E_{G'}$.
\item If $\mu_N(s)$ is defined and \texttt{:Observation} is the type of $s$,  and $p$ is not {\tt \footnotesize :result} and {\tt \footnotesize :samplingTime}, then $(\mu_N(s) \;p\;o)$ is in $E_{G'}$.
\item If $p$ is {\tt \footnotesize :samplingTime}, then $(s \;p\;o)$ is in $E_{G'}$.
\item If $\mu_N(s)$ and $\mu_N(o)$ are defined and $p$ is {\tt \footnotesize :result}, then $(\mu_N(s) \;p\;\mu_N(o))$ and $(s \;p\;o)$ are in $E_{G'}$.
\item If $\mu_N(s)$ is defined and type of $s$ is {\tt \footnotesize :MeasureData}, then $(\mu_N(s) \;p\;o)$ is in $E_{G'}$.
\item Otherwise, the RDF triple $t$ is preserved in $E_{G'}$.
\end{itemize}

\item Multiplicity of measurements is reduced, i.e., for all $val$, $uom$ such that $M_m(val,uom|G)\geq$ 1, then $M_m(val,uom|G')$=1, and
\item Multiplicity of observations is reduced, i.e., for all $proc$, $ph$, $pp$, $val$, $uom$, such that, $M_o(proc,ph,pp,val,uom|G)\geq$ 1, then $M_o(proc,ph,pp,val,uom|G')$=1.
\end{itemize}
\end{definition}

\begin{figure*}[tb]
\centering
    \subfloat[Compact Molecules]{
      \includegraphics[width=.25\linewidth]{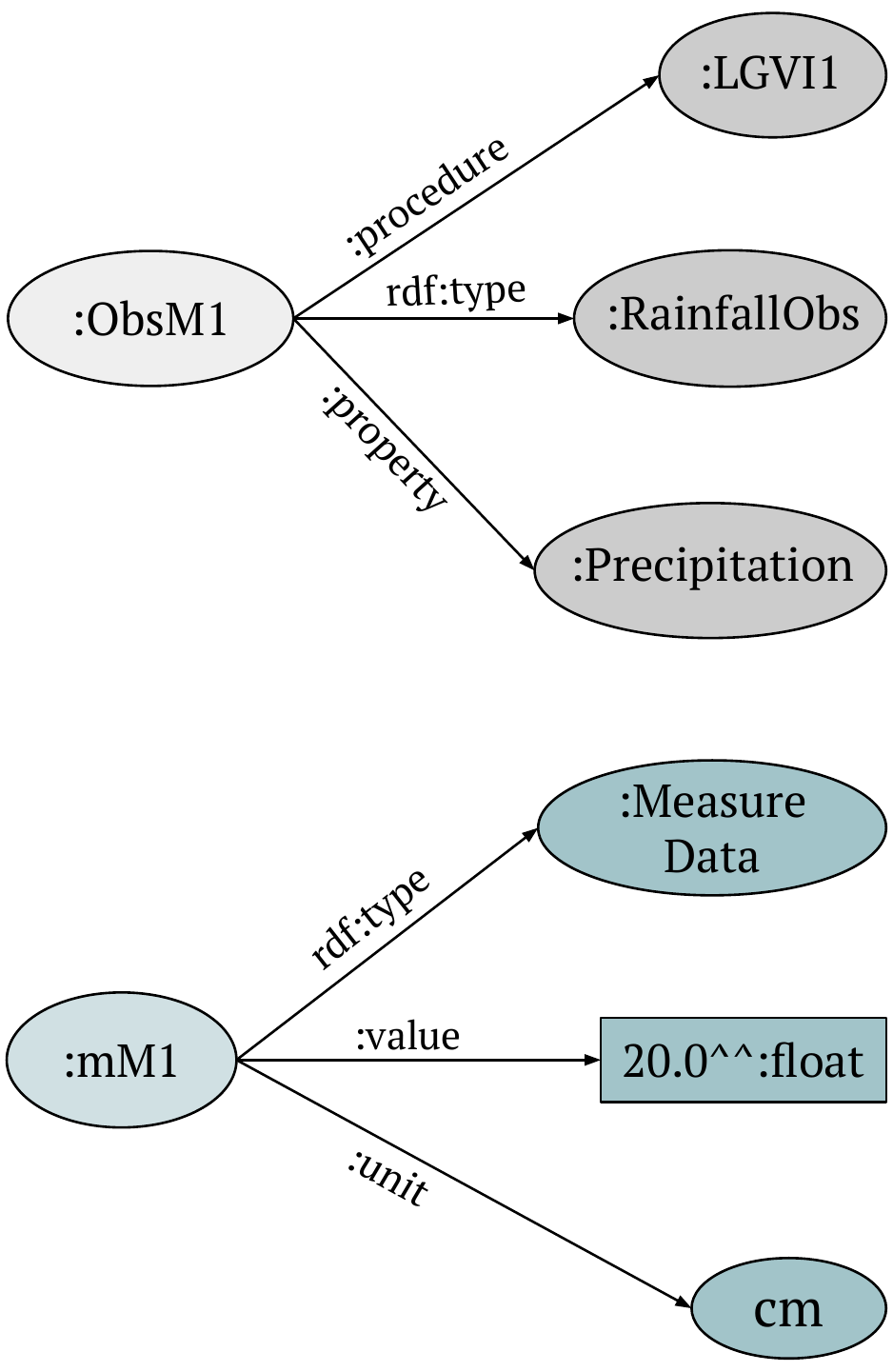}
      \label{fig:compMol}}
       \subfloat[Entity mappings $\mu_N$ from $G$ into $G'$]{
      \includegraphics[width=.4\linewidth]{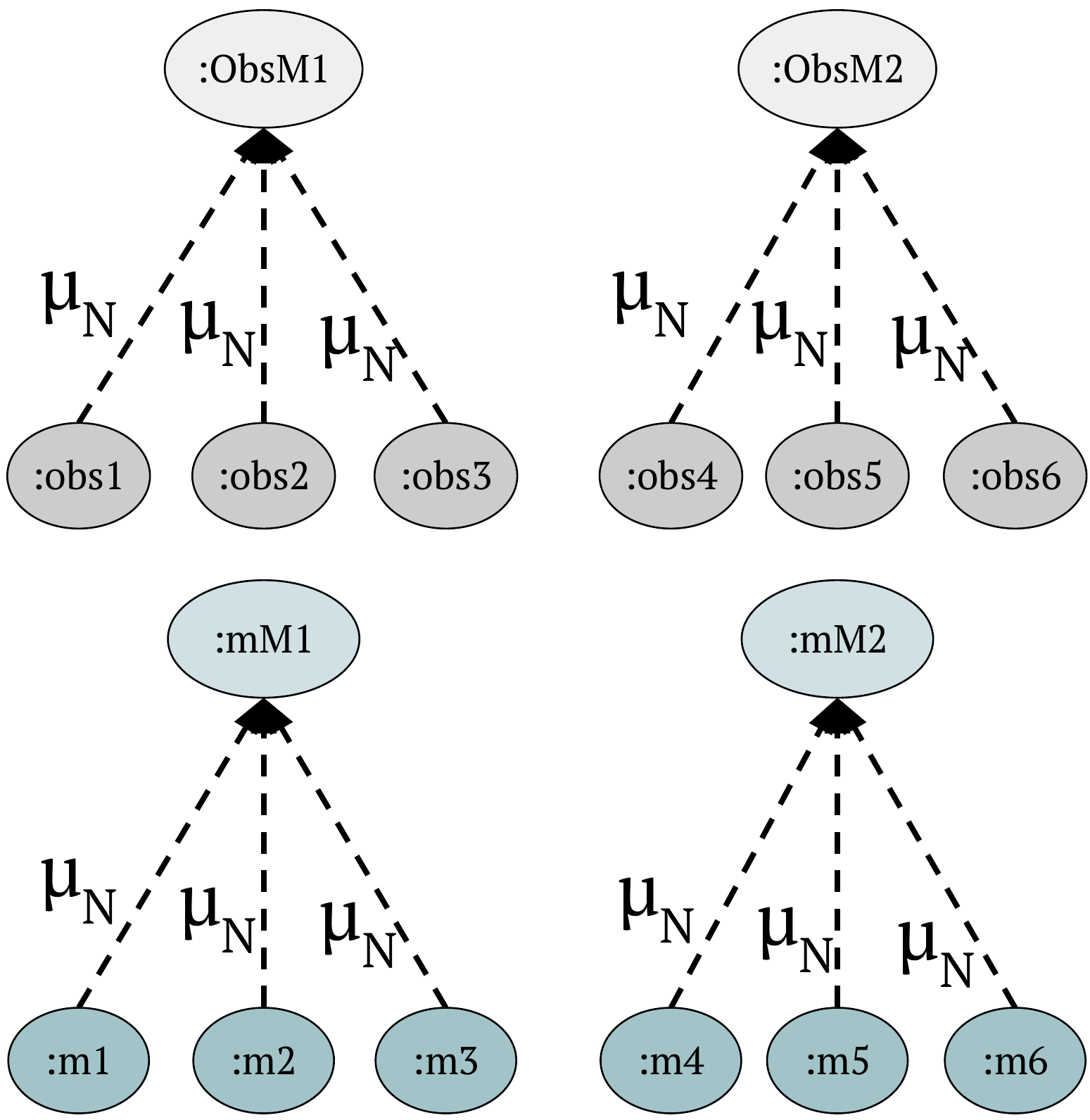}
      \label{fig:homomorphism}}
       \\
       \vspace{5pt}
    \subfloat[Factorized RDF graph $G'$]{
      \includegraphics[width=.6\linewidth]{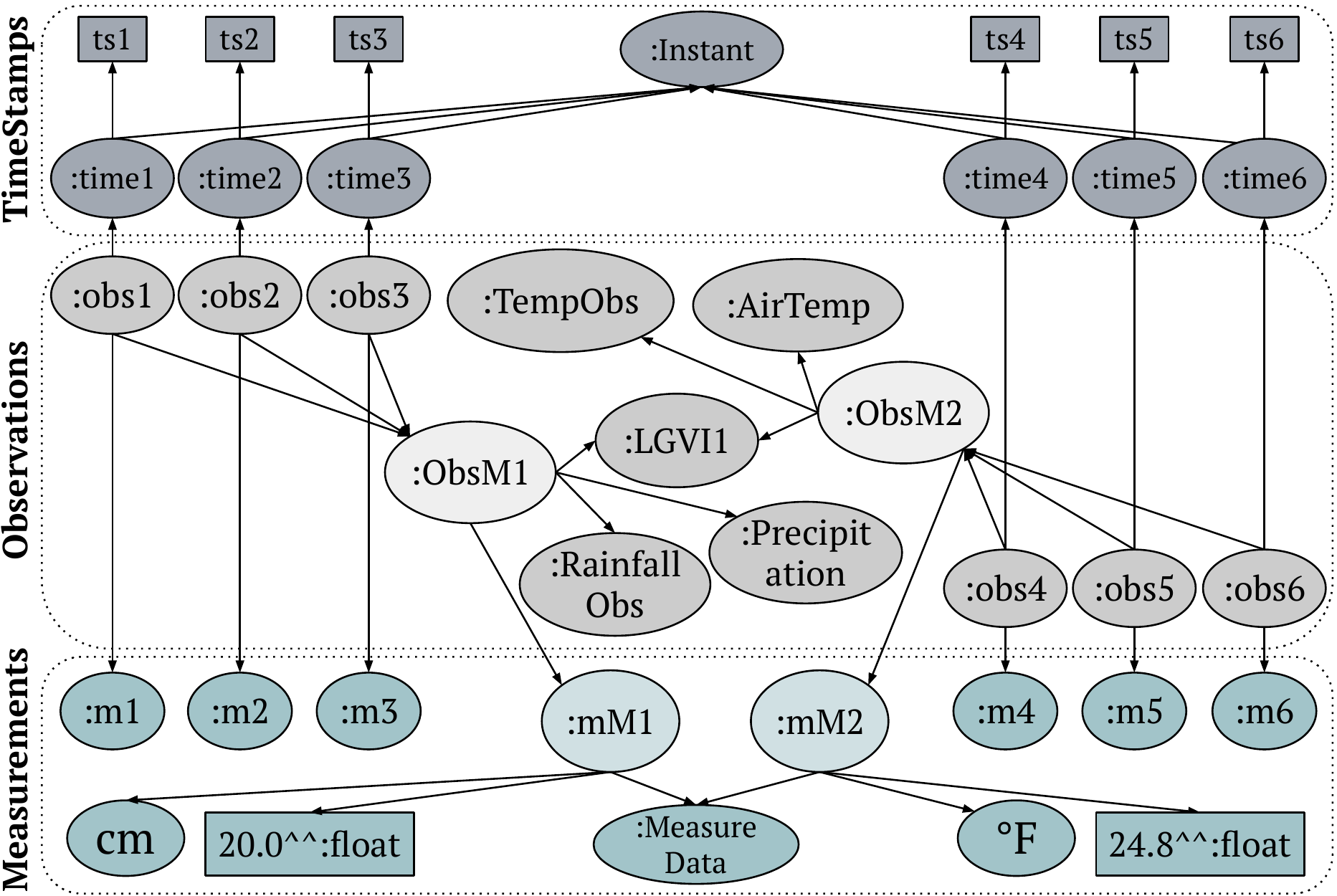}
      \label{fig:factorizedGraph}}
    \caption{{\bf Instance of the Semantic Sensor Data Factorization Problem.} Factorized RDF graph $G'$ of $G$ in Figure~\ref{fig:multiplicityExp}. (a) Compact observation and measurement molecules are presented. (b) Entity mappings $\mu_N$ from graph in \autoref{fig:multiplicityExp} %, that maps \texttt{:m1}, \texttt{:m2}, and \texttt{:m3} to \texttt{:mM1}, and \texttt{:m4}, \texttt{:m5}, and \texttt{:m6} to \texttt{:mM2}, and \texttt{:obs1},  \texttt{:obs2}, and \texttt{:obs3} to \texttt{:oM1}, and \texttt{:obs4},  \texttt{:obs5}, and \texttt{:obs6} to \texttt{:oM2}
    to surrogate entities; (c) Factorized RDF graph $G'$ with  multiplicities equal to one. 
    %Mappings between the observations and the surrogate observation are explicitly stated, while the mappings between measurements and the surrogate measurements are represented through a path in the RDF graph.
    }
\end{figure*}

%\begin{definition}[The SSDF Problem]
%Given an RDF graph $G=(V_G,E_G,L_G)$ representing sensor data with the SSN ontology, the problem of \textit{semantic sensor data factorization (SSDF)}  in $G$, corresponds to finding a \textit{factorized RDF graph} $G'=(V_{G'},E_{G'},L_{G'})$ of $G$.
%\end{definition}
Given an RDF graph $G=(V_G,E_G,L_G)$ representing sensor data with the SSN ontology, the problem of \textit{semantic sensor data factorization (SSDF)}  in $G$, corresponds to finding a \textit{factorized RDF graph} $G'=(V_{G'},E_{G'},L_{G'})$ of $G$.
\begin{table*}[htbp] 
			\centering 
			%\scriptsize
			\caption[Query Rewriting Rules]{{\bf Query Rewriting Rules}. The rewriting rules for observations and measurements with respect to the relevant properties are expressed in terms of triple patterns. The variables representing observations and measurements are replaced in SPARQL query clauses, i.e., SELECT, ORDER BY, GROUP BY, and FILTER.}
			%\resizebox{\columnwidth}{!}{%
			\resizebox{\textwidth}{!}{%
			\begin{tabular}{|c| | l | l |}
				\hline
				\textbf{Rule Name} &\multicolumn{1}{c}{\textbf{Head}} & \multicolumn{1}{|c|}{\textbf{Body}} \\  \hline
				fssn1 & $?obs \; \textit{rdf:type} \; \textit{:Observation}$ &$?obs\;\;\;\;\; \textit{rdf:type} \;\;\;\;\;\;\;\;\;\;\;\;\; \textit{:Observation}$\\
				& & $?Xobs \; \textit{:observationOf} \;\; ?obs $    \\ 
				& & Replace $?obs$ by $?Xobs$ in query clauses \\
				\hline
				fssn2 & $?obs \; \textit{:procedure} \; ?sensor$ & $?obs\;\;\;\;\; \textit{:procedure} \;\;\;\;\;\;\;\; ?sensor$ \\
				&& $?Xobs \; \textit{:observationOf} \;\; ?obs$    \\ 
				& & Replace $?obs$ by $?Xobs$ in query clauses \\
				\hline
				fssn3 & $?obs \; \textit{:property} \; ?property$ & $?obs\;\;\;\; \textit{:property} \;\;\;\;\;\;\;\;\;\;\; ?property$ \\
				&& $?Xobs \; \textit{:observationOf} \;\; ?obs $    \\ 
				& & Replace $?obs$ by $?Xobs$ in query clauses \\
				\hline
				fssn4 & $?m \; \textit{rdf:type} \; \textit{:MeasureData}$ & $?m\; \;\;\;\;\;\; \textit{rdf:type} \;\;\;\;\;\;\;\;\;\;\;\;\; \textit{:MeasureData}$ \\
				&& $?Xobs \;\; \textit{:observationOf} \;\; ?obs$    \\ 
				& & $?Xobs\; \; \textit{:result} \;\;\;\;\;\;\;\;\;\;\;\;\;\;\; ?Xm$ \\
				&& Replace $?m$ by $?Xm$ in query clauses \\
				\hline
				fssn5 & $?m \; \textit{:value} \; ?val$ & $?m\; \;\;\;\;\; \textit{:value} \;\;\;\;\;\;\;\;\;\;\;\;\;\;\;\; ?val$\\
				&& $?Xobs \; \textit{:observationOf} \;\; ?obs$    \\ 
				& & $?Xobs\;  \textit{:result} \;\;\;\;\;\;\;\;\;\;\;\;\;\;\; ?Xm$ \\
				&& Replace $?m$ by $?Xm$ in query clauses \\
				\hline
				fssn6 & $?m \; \textit{:unit} \; ?uom$ & $?m\; \;\;\;\;\; \textit{:unit} \;\;\;\;\;\;\;\;\;\;\;\;\;\;\;\;\; ?uom$ \\
				&& $?Xobs \; \textit{:observationOf} \;\; ?obs$    \\ 
				& & $?Xobs\;  \textit{:result} \;\;\;\;\;\;\;\;\;\;\;\;\;\;\; ?Xm$ \\
				&& Replace $?m$ by $?Xm$ in query clauses \\
				\hline
				fssn7 & $?obs \; \textit{:result} \; ?m$ & $?obs\; \;\;\; \textit{:result} \;\;\;\;\;\;\;\;\;\;\;\;\;\;\;\; ?m$ \\
				&& $?Xobs \; \textit{:observationOf} \;\; ?obs$    \\ 
				& & $?Xobs\;  \textit{:result} \;\;\;\;\;\;\;\;\;\;\;\;\;\;\; ?Xm$\\
				&&Replace $?obs$ by $?Xobs$ and $?m$ by $?Xm$ in query clauses\\
				\hline
			\end{tabular}}
			\label{tab:entailmentRules}    
		\end{table*}
Consider RDF graphs $G$ and $G'$ in ~\autoref{fig:multiplicityExp} and \autoref{fig:factorizedGraph}, respectively. 
Furthermore,  \autoref{fig:homomorphism} shows  mappings $\mu_N$ that assign measurements \texttt{:m1},  \texttt{:m2}, and \texttt{:m3} in $G$ to the surrogate measurement \texttt{:mM1} in $G'$, and \texttt{:m4},  \texttt{:m5}, and \texttt{:m6} to \texttt{:mM2}. Similarly, \texttt{:obs1}, \texttt{:obs2}, and \texttt{:obs3} are mapped to \texttt{:obsM1}, and \texttt{:obs4}, \texttt{:obs5}, and \texttt{:obs6} to \texttt{:obsM2}; $\mu_N$ is the identity for the rest of the nodes.
Measurement and observation multiplicities are one in $G'$, which is the \textit{factorized graph} of $G$. 

Once RDF graphs are factorized, query processing is performed against the factorized graphs. SPARQL queries over the original RDF graphs need to be re-written against the corresponding factorized RDF graphs in the way that equivalent answers are computed.
We have defined seven query rewriting rules, given in Table~\ref{tab:entailmentRules}. Each rule is given a name, i.e., \textit{fssn1}, \textit{fssn2}, \textit{fssn3}, \textit{fssn4}, \textit{fssn5}, \textit{fssn6}, and \textit{fssn7}, and has a head and a body. The head of a rule corresponds to the triple pattern in the query against original RDF graph, whereas the body of the rule represents the corresponding triple patterns against the factorized RDF graph. The head of the rule \textit{fssn1} contains a triple pattern that matches to all the original observations, whereas the body of the rule matches the corresponding surrogate observations. Moreover, the variable substitutions, i.e., $?obs$ by $?Xobs$, are maintained for the query clauses such as SELECT, FILTER, GROUP BY etc. 
The head of rule \textit{fssn2}, matches the procedure generating the original observations, while the body of the rule matches the procedure of the surrogate observations, and keeps the variable substitutions.
Similarly, the head of the rule \textit{fssn3} consists of a triple pattern that matches the observed property in the original RDF graph, and the body of the rule extracts the observed property of the surrogate observations. 

The rules \textit{fssn4}, \textit{fssn5}, and \textit{fssn6} are used to rewrite the triple patterns involving the measurement properties. 
The head of the rule \textit{fssn4} contains a triple pattern matching all the entities of measurements in original RDF graphs. The body of the rule \textit{fssn4} contains three triple patterns that match to the surrogate measurements, as well as, associate the original observations with the surrogate observations, using property \textit{:observationOf}, and relate observations to corresponding measurements in the original RDF graph using property \textit{:result}. Moreover, the body of the rule maintains the measurement variable substitutions. 
The head of the rule \textit{fssn5} find matches of the values of measurements in original RDF graphs, whereas the body of the rule find values of the surrogate measurements in factorized RDF graphs. Further, the triple patterns extract associations between the original and surrogate observations, as well as between the original observations and corresponding original measurements, and maintain the measurement variable substitutions.
The head of rule \textit{fssn6} contains the triple pattern matching the measurement units in original RDF graph, whereas the body matches the unit of the surrogate measurements. Also, body maintains associations between original and surrogate observations and original observations and corresponding measurements along with measurement variable substitutions. 
Finally, the head of the rule \textit{fssn7} maps the original observations and the measurements in original RDF graphs using property \textit{fssn7}. 
The body of the rule \textit{fssn7} find associations between surrogate observations and surrogate measurements in factorized RDF graphs using property \textit{:result}.
Likewise associations between the original and surrogate observations and the original observations and original measurements are maintained. Additionally, the variable substitutions for the observations and measurements are maintained.

Let $G$ and $G'$ be RDF graphs such that $G'$ is a factorized graph of $G$. Consider a SPARQL query $Q$ over $G$.
The problem of \textit{evaluating SPARQL queries against a factorized RDF graph} corresponds to the problem of transforming $Q$ into a SPARQL query $Q'$ over $G'$ such that the results of evaluating $Q$ over $G$ and the results of $Q'$ over $G'$ are the same, i.e.,  the condition $[[Q]]_{G}=[[Q']]_{G'}$ is satisfied.

 An instance of the problem of evaluating queries on factorized RDF graphs is shown in \autoref{fig:queryExecution}. A SPARQL query $Q$ over the RDF graph $G$ in \autoref{fig:multiplicityExp} is presented in Figure~\ref{fig:queryOriginal}. The SPARQL query $Q'$ in \autoref{fig:queryFactorized}, corresponds to a rewriting of $Q$, against $G'$ which represents factorization of $G$. The evaluations of $Q$ and $Q'$ 
produce the same answers. In this work, we present SPARQL query rewriting rules that allow for rewriting a query $Q$ into a query $Q'$.

\begin{figure*}[ht!]
\centering
 \vspace{0pt}\subfloat[SPARQL Query over Original RDF Graph]{
      \includegraphics[width=.7\linewidth]{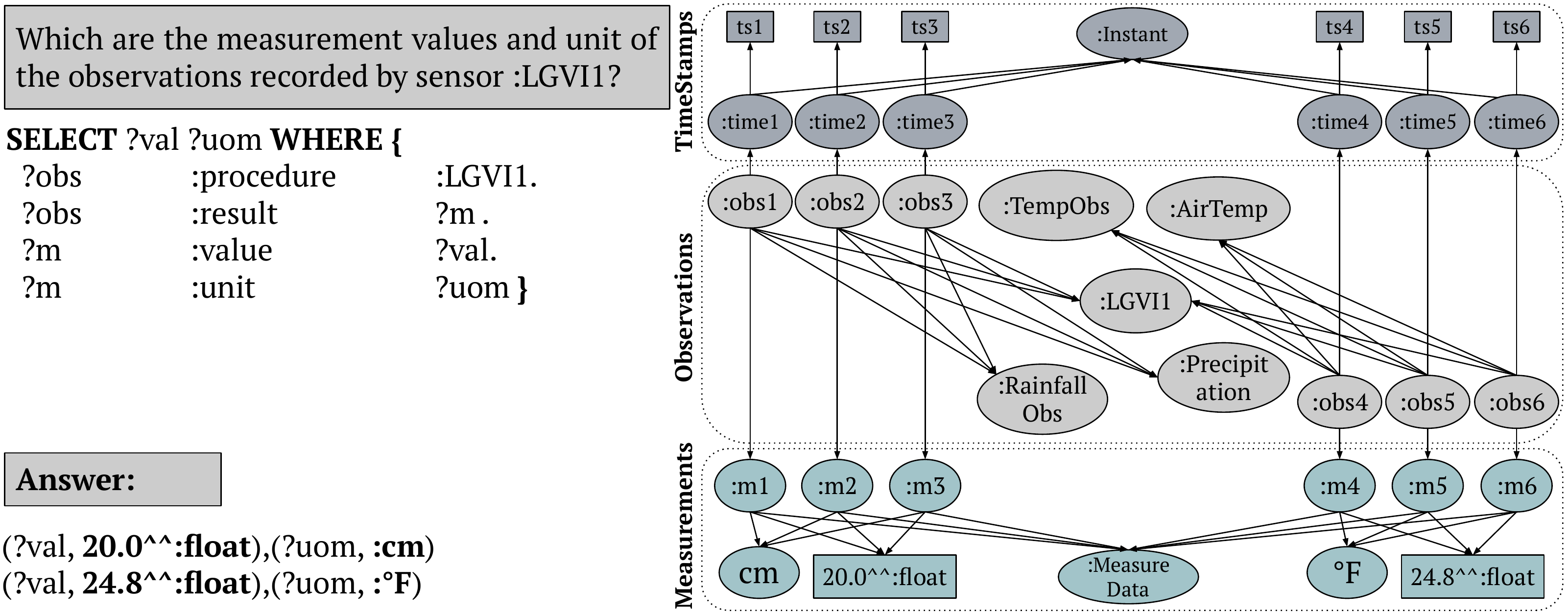}
      \label{fig:queryOriginal}}  
     \\
       \vspace{5pt}\subfloat[SPARQL Query over Factorized RDF Graph]{
      \includegraphics[width=.7\linewidth]{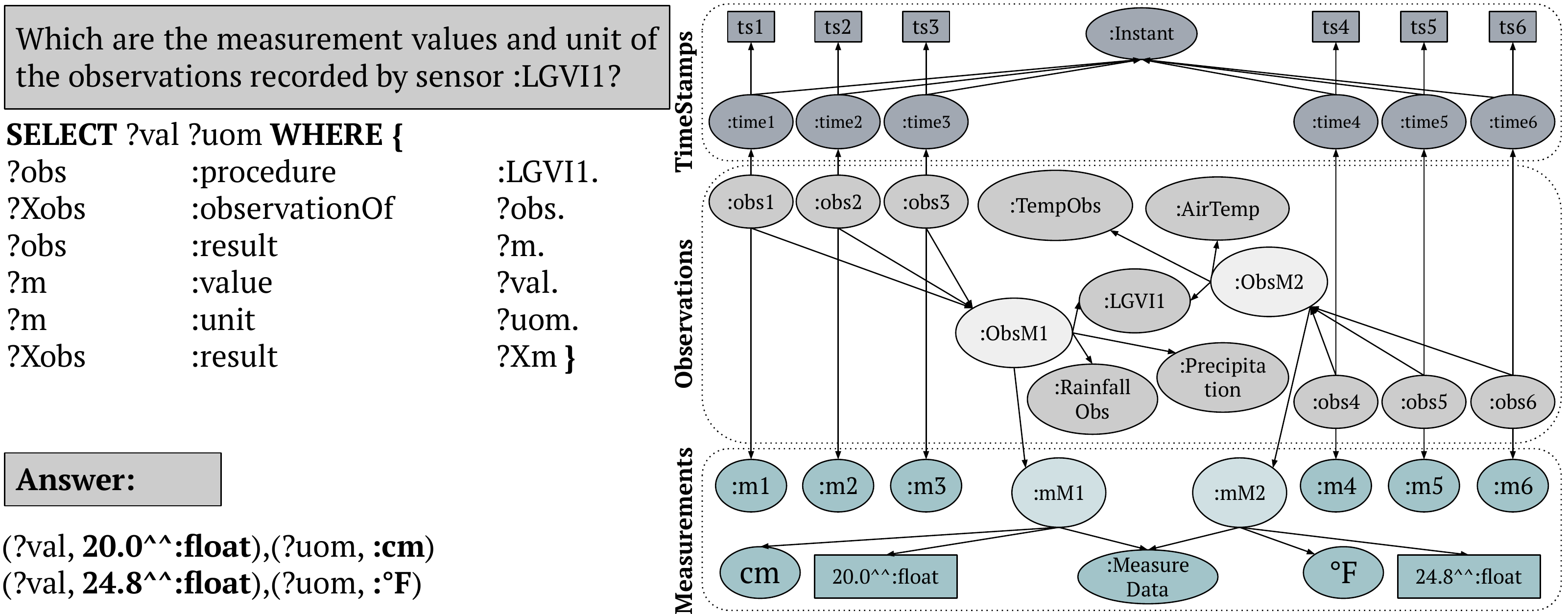}
      \label{fig:queryFactorized}} 
    \caption{{\bf Instance of the Query Evaluation Problem}. Evaluation of SPARQL queries over the original and factorized RDF graphs respects set semantics. (a) SPARQL query over original RDF graph selects the values and unit collected by \textit{:LGVI1}; (b) Rewritten SPARQL query over factorized RDF graph, in Figure~\ref{fig:factorizedGraph}.}
     \label{fig:queryExecution}
\end{figure*}

\subsection{A Factorization Approach}

We present a solution to the \textit{semantic sensor data factorization (SSDF)} problem.% factorizing RDF graphs describing semantic sensor data. 
A sketch of the proposed factorization approach is presented in Algorithm~\ref{algo:factorization}.
The algorithm receives an RDF graph $G(V_G,E_G,L_G)$, and a graph $G''(V_{G''},E_{G''},L_{G''})$ representing an already factorized graph and the entity mappings $\mu_{N''}$ in $G''$.  The algorithm incrementally generates a factorized RDF graph $G'(V_{G'},E_{G'},L_{G'})$ of $G$, and the entity mappings $\mu_N$ from the observations and measurements in $G$ to the surrogate observations and measurements in $G'$, respectively. The algorithm initializes a set of entity mappings $\mu_N$ and the sets of nodes $V_{G'}$, edges $E_{G'}$, and labels $L_{G'}$ of the factorized graph $G'$ (line 1).
If for all the observations with observed phenomenon $ph$, sensor procedure $proc$, observed property $pp$, and for the corresponding measurements with value $val$ and the related unit $uom$ in $G$, 
%the algorithm (lines 2-9) checks for the existence of the corresponding surrogate observation and measurement in $G'$. If 
the surrogate observation and measurement are already in $G''$, then the observations and measurements in $G$ are mapped in $\mu_N$ to the surrogate observation and measurement in $G''$, respectively(lines 2-9). Furthermore, observations in $G$ are linked using the property {\tt :observationOf} to the surrogate observations in $\mu_N$(line 6-8). 
If the surrogate observation and measurement are not in $G''$, the algorithm (lines 11) creates corresponding surrogate entities in $G'$, i.e., the subjects of compact measurement and observation molecules are created.
In lines 12-13, the algorithm maps all the measurements, related to $val$ and $uom$ in $G$, to the surrogate measurements in $\mu_N$.
For all the observations with observed phenomenon $ph$, sensor procedure $proc$, observed property $pp$, measurement value $val$ and unit of measurement $uom$ in $G$, adds in $\mu_N$ the mappings of all the observations in $G$ with the surrogate observations in $G'$ in lines 14-15.
%Once all the mappings of the measurements and observations in $G$ to the corresponding surrogate measurements and observations, respectively, in $G'$ are in $\mu_N$, the factorized graph $G'$ is created using $\mu_N$ (lines 10-20).
Once the mappings are in $\mu_N$, the nodes and edges representing the mapped observations and measurements in $G$ are processed. All nodes $s$ and $o$ related to the property {\tt :result} in $G$ are added to $G'$ along with their associations. Moreover, a new edge relating $s$ and $\mu_N(s)$ using the property {\tt :observationOf} is added to $G'$ (lines 18-20).
If $s$ and $o$ are linked using a property {\tt rdf:type} and $o$ is either {\tt :Observation} or {\tt :MeasureData}, then a new edge ($\mu_N(s)\; p\; o$) is added to $G'$ along with $\mu_N(s)$ and $o$ (lines 21-22).
If $s$ and $o$ are associated through a predicate $p$ in \{:procedure, :property, :value, :unit\}, then a new edge ($\mu_N(s)\; p\; o$) is added to $G'$ in lines 23-24.
Otherwise, the edge $(s\; p\; o)$ is added to the $G'$ in lines 25-26.

\begin{algorithm}[!htbp]
\begin{scriptsize}
\caption{The Incremental Factorization Algorithm}
\label{algo:factorization}
\KwIn{An RDF graph $G(V_G,E_G,L_G)$, Previously factorized RDF Graph $G''(V_{G''},E_{G''},L_{G''})$, and entity mappings $\mu_{N''}$} 
\KwOut{ Factorized RDF Graph $G'(V_{G'},E_{G'},L_{G'})$, and entity mappings $\mu_N$ } 
\DontPrintSemicolon
$\mu_N \longleftarrow \mu_{N''}, V_{G'} \longleftarrow V_{G''}, E_{G'} \longleftarrow E_{G''}, L_{G'} \longleftarrow L_{G''}$\;
\ForAll{$proc, ph, pp, val, uom \in V_G$ such that $SO=
\{obs|(obs\; \texttt{rdf:type}\; ph) \in G,$
$(obs \;\texttt{:procedure} \; proc) \in G,$ 
$(obs\; \texttt{:property}\; pp) \in G, 
(obs\; \texttt{:result}\; m) \in G, 
(m \;\texttt{rdf:type} \;\texttt{:MeasureData}) \in G,
(m\; \texttt{:unit}\; uom) \in G, 
(m\; \texttt{:value}\; val) \in G\}$, and $SM=\{m|(m\; \texttt{rdf:type}\; \texttt{:MeasureData}) \in G, (m\; \texttt{:unit}\; uom) \in G,
(m\; \texttt{:value}\; val) \in G \}$}{
\If{$\exists mM$, $oM$ such that $
(mM \;\texttt{rdf:type} \;\texttt{:MeasureData}) \in G'', $
$(mM\; \texttt{:unit}\; uom) \in G'', $
$(mM\; \texttt{:value}\; val) \in G'',
(oM\; \texttt{rdf:type}\; ph) \in G'',
(oM \;\texttt{:procedure} \; proc) \in G'',
(oM\; \texttt{:property}\; pp) \in G'',$ and 
$(oM\; \texttt{:result}\; mM) \in G''$
}{
\ForEach{$(s\; \texttt{rdf:type}\; o) \in E_G \land s, o \in V_G \land  \texttt{rdf:type} \in L_G$ such that $s \in SM \cup SO $}{
\If{$s \in SM$}{
$\mu_N\leftarrow\mu_N \cup \{(s,mM)\}$
}
\Else{$\mu_N \leftarrow \mu_N \cup \{(s,oM)\}, E_{G'} \leftarrow E_{G'} \cup (s\;\texttt{:observationOf}\; \mu_N(s))\}$\;
     $V_{G'} \leftarrow V_{G'} \cup \{s,o\}$,
     $L_{G'} \leftarrow L_{G'} \cup \{\texttt{:observationOf}\}$
}
}
}
\Else{
$mM \leftarrow SurrogateMeasurement(), oM \leftarrow SurrogateObservation()$\;
\ForEach{$m \in SM$}{
$\mu_N\leftarrow\mu_N \cup \{(m,mM)\}$ \\
}
\ForEach{$obs \in SO$}{
$\mu_N \leftarrow \mu_N \cup \{(obs,oM)\}$\;
}

}
\ForEach{$(s\; p\; o) \in E_G \land s, o \in V_G \land p \in L_G$ }{
\If{$s \in SM \cup SO $}{
\If{$p==\texttt{:result}$}{
    $E_{G'} \leftarrow E_{G'} \cup \{(s\;p \;o), (\mu_N(s)\;p\;\mu_N(o)), (s\;\texttt{:observationOf}\; \mu_N(s))\}$\;
     $V_{G'} \leftarrow V_{G'} \cup \{s,o,\mu_N(s),\mu_N(o)\}$,
     $L_{G'} \leftarrow L_{G'} \cup \{p,\texttt{:observationOf}\}$
  }
    \ElseIf{$p==\texttt{rdf:type}\; \&\&\; (o==\texttt{:Observation}  ||  o==\texttt{:MeasureData})$}{ 
     $E_{G'} \leftarrow E_{G'} \cup \{(\mu_N(s)\; p\; o)\}$, $V_{G'} \leftarrow V_{G'} \cup \{\mu_N(s),o\}$, $L_{G'} \leftarrow L_{G'} \cup \{p\}$
  }
    \ElseIf{$ p==\texttt{:procedure}  ||  p==\texttt{:property}  ||  p==\texttt{:value} ||  p==\texttt{:unit}$}{ 
      $E_{G'} \leftarrow E_{G'} \cup \{(\mu_N(s)\; p\; o)\}, V_{G'} \leftarrow V_{G'} \cup \{\mu_N(s),o\}$, $L_{G'} \leftarrow L_{G'} \cup \{p\}$
  }
  \Else{
    $E_{G'} \leftarrow E_{G'} \cup \{s\; p\; o)\}$, $V_{G'} \leftarrow V_{G'} \cup \{s,o\}$, $L_{G'} \leftarrow L_{G'} \cup \{p\}$
  }
}
 %\Else{
    %$E_{G''} \leftarrow E_{G''} \cup \{s\; p\; o)\}$, $V_{G''} \leftarrow V_{G''} \cup \{s,o\}$, $L_{G''} \leftarrow L_{G''} \cup \{p\}$}
}
}
\Return {$G'(V_{G'},E_{G'},L_{G'}),\mu_N$ }
\end{scriptsize}
\end{algorithm}

Figure~\ref{fig:transformationInstance} depicts a portion of the RDF in Figure~\ref{fig:multiplicityExp} and the corresponding transformation in the factorized RDF graph in ~\autoref{fig:factorizedGraph}.
The surrogate measurements and observations, and the new edges are highlighted in bold. The Algorithm~\ref{algo:factorization} creates the surrogate measurements and observations in  line 3 and 7; new edges are created in line 12, 15, and 17.
Additionally, assumptions about the characteristics of the associations between the nodes in the graph are presented. While some edges existing in the RDF graph in \autoref{fig:multiplicityExp} are not present in the factorized RDF graph, these associations can be obtained by traversing the graph through the surrogate observations and measurements. The implicit satisfaction of all the associations in the original RDF graph that are not included in the factorized graph is restricted under the following assumptions: 

\begin{figure*}[ht!]
\centering
 \subfloat[Query Rewriting]{
      \includegraphics[width=.7\linewidth]{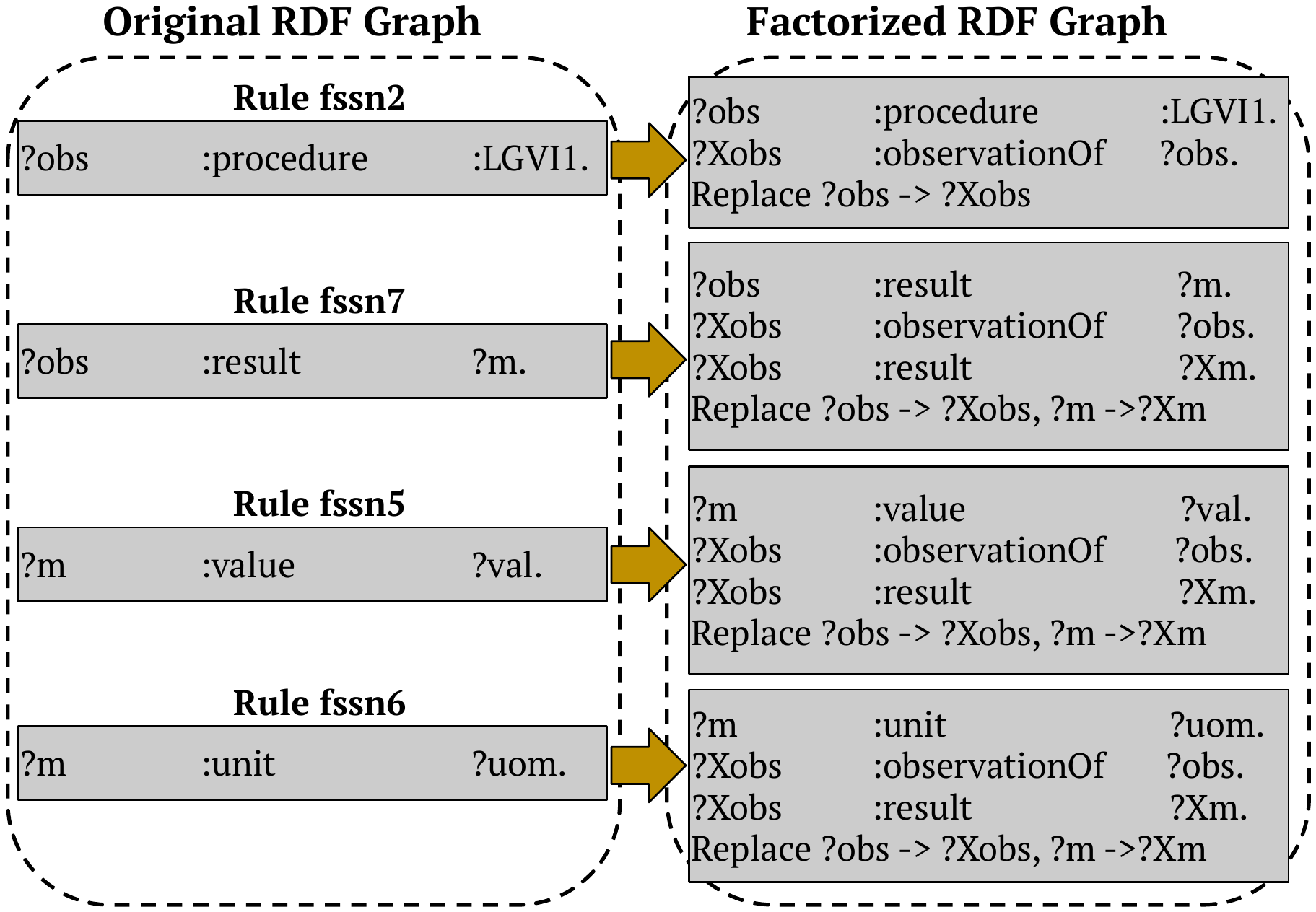}
      \label{fig:transformations}}  
      \\
       \vspace{5pt}\subfloat[Original and Factorized RDF Graphs]{
      \includegraphics[width=.7\linewidth]{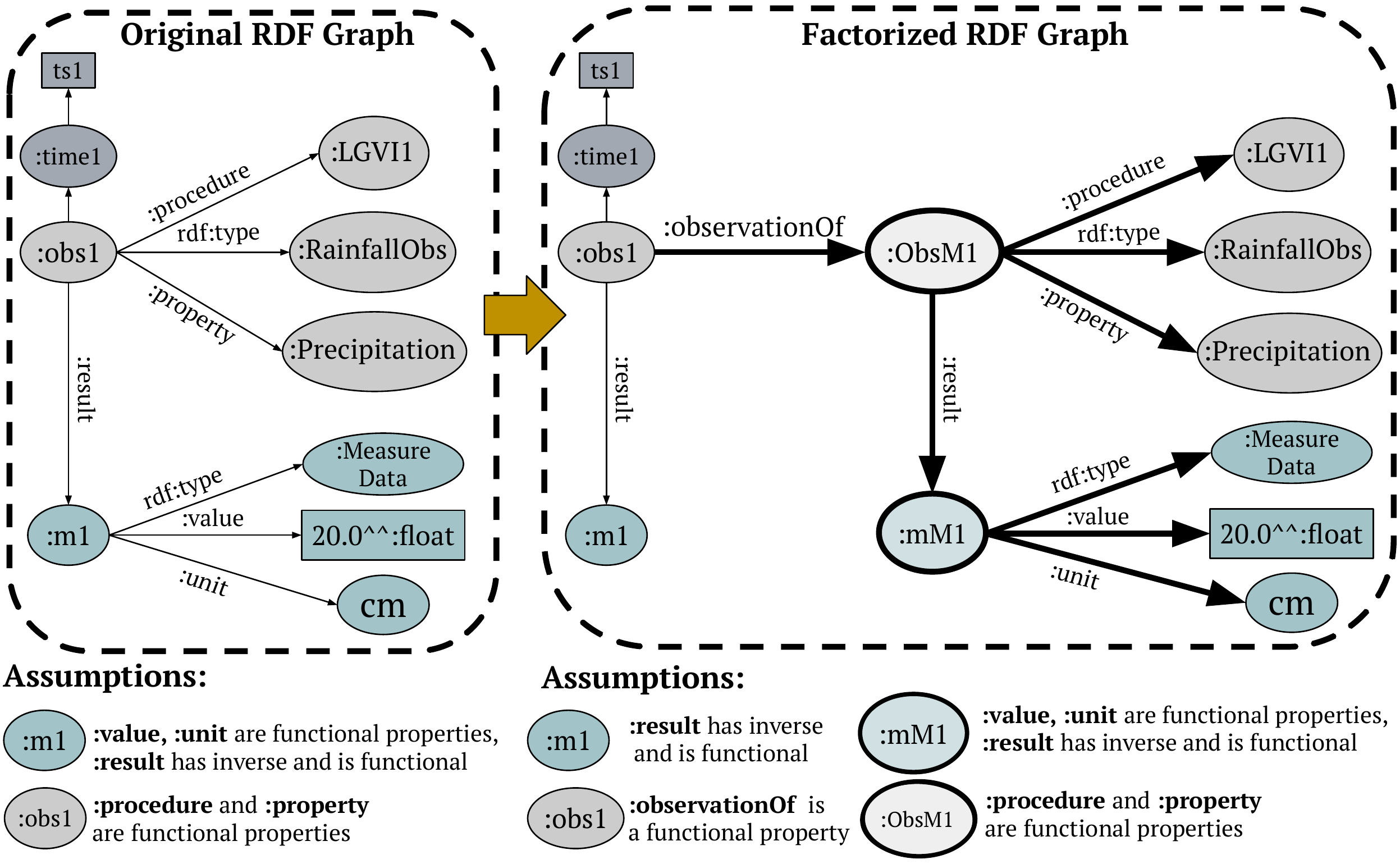}
      \label{fig:transformationInstance}} 
    \caption{{\bf Example of Query Rewriting}. Query rewriting rules are presented. (a) Query rewriting rules from Table~\ref{tab:entailmentRules} are used to rewrite the query in Figure~\ref{fig:queryOriginal} into the query in Figure~\ref{fig:queryFactorized}. (b) Portions of the RDF graphs (original and factorized). Nodes and edges highlighted in bold are added during the RDF graph factorization.
    }
    \label{fig:GraphFactorization}
\end{figure*}

For all observations \texttt{:obs} and measurements \texttt{:m} in $G$, the following hold. 
\begin{itemize}
\item \textbf{Measurement}: the properties \texttt{:value} and \texttt{:unit} of measurement are both functional properties for any measurement \texttt{:m}. Furthermore, the property \texttt{:result} that associates an observation with a measurement has a functional inverse.  
\item \textbf{Observation}: the property \texttt{:procedure} that associates an observation and a procedure is a functional property for any observation \texttt{:obs}.
\end{itemize}

The following hold for a surrogate observation \texttt{:obsM}, a surrogate measurement \texttt{:mM}, an observation \texttt{:obs}, and a measurement \texttt{:m} in $G'$.

\begin{itemize}
\item \textbf{Surrogate Observation}: The properties \texttt{:procedure} and \texttt{:property} are functional properties for any surrogate observation \texttt{:obsM}. 
\item \textbf{Surrogate Measurement}: \texttt{:value} and \texttt{:unit} are both functional properties for any surrogate measurement \texttt{:mM}. Furthermore, \texttt{:result} that associates a surrogate observation with a surrogate measurement has a functional inverse.  
\item \textbf{Observation}: \texttt{:observationOf} property that associates an observation \texttt{:obs} with a surrogate observation is a functional property.
\item \textbf{Measurement}: \texttt{:m} is related to only one observation, i.e., 
\texttt{:result} associates an observation with a measurement, and has a functional inverse.  
\end{itemize}  

We are assuming that SPARQL queries against the original and factorized RDF graphs are evaluated under the set semantics, i.e., no duplicates are in the answers. 
Coming back to the motivating example, \autoref{fig:motivatingExampleFactorization} illustrates the factorized RDF graph of the graph in \autoref{fig:motivatingExample}.
The factorized RDF graph in \autoref{fig:FactorizedRDFGraph} is sparse and the average number of neighbors has been reduced from 6.4 to 2.5.
This indicates that the number of RDF triples describing an observation is reduced after factorization.
\autoref{fig:factorizedNT} shows that for each measurement value the number of associated RDF triples in the factorized RDF graph is reduced by 74\%.

\subsection{Queries over Factorized RDF Graphs}
In this section, we define the algorithm that solves the problem of query evaluation on a factorized RDF graph. 
Table~\ref{tab:entailmentRules} presents the rules to rewrite a SPARQL query against an original SSN RDF graph into a query against the corresponding factorized RDF graph.
The query rewriting rules are defined in terms of SPARQL triple patterns.
For each property of the observation and measurement classes, a rewriting rule is defined. Furthermore, substitutions for the observation and measurement variables in the query clauses, i.e., SELECT, ORDER BY, GROUP BY, and FILTERS etc, are presented.
Given a SPARQL query and a set $R$ of query rewriting rules, the Algorithm~\ref{algo:query} describes the steps performed to each set of triple patterns that composes a Basic Graph Pattern (BGP). If the input query consists of several BGPs, the structure of the original query remains the same, and Algorithm~\ref{algo:query} is applied to each BGP within the query using rules in \autoref{tab:entailmentRules}. 

\begin{algorithm}[!ht]
%\begin{scriptsize}
\caption{The Query Rewriting Algorithm}
\label{algo:query}
\KwIn{Set $ST$ of triple patterns in a BGP of $Q$ and set $SR$ of query rewriting rules} 
\KwOut{$ST_{new}$ the rewriting of $ST$ under $SR$} 
\DontPrintSemicolon
$ST_{new} \longleftarrow \emptyset$\\
\ForEach{$t \in ST$}{
\begin{itemize}
    \item Select $r \in SR$ such that $t$ matches the head\\ of $r$ and instantiate the body of $r$
    \item Let $SQ_t$ be the matched body of $r$ and $variableSubstitutions$ be the set of mapp-\\ings between variables in $t$ into $SQ_t$, add ($t, SQ_t, variableSubstitutions$) to $ST_{new}$
\end{itemize}
}
\Return {\tt $ST_{new}$}
%\end{scriptsize}
\end{algorithm}

\begin{figure*}[ht!]
\centering
\subfloat[Fact. {\bf RDF Graph} ]{\includegraphics[width=.25\linewidth]{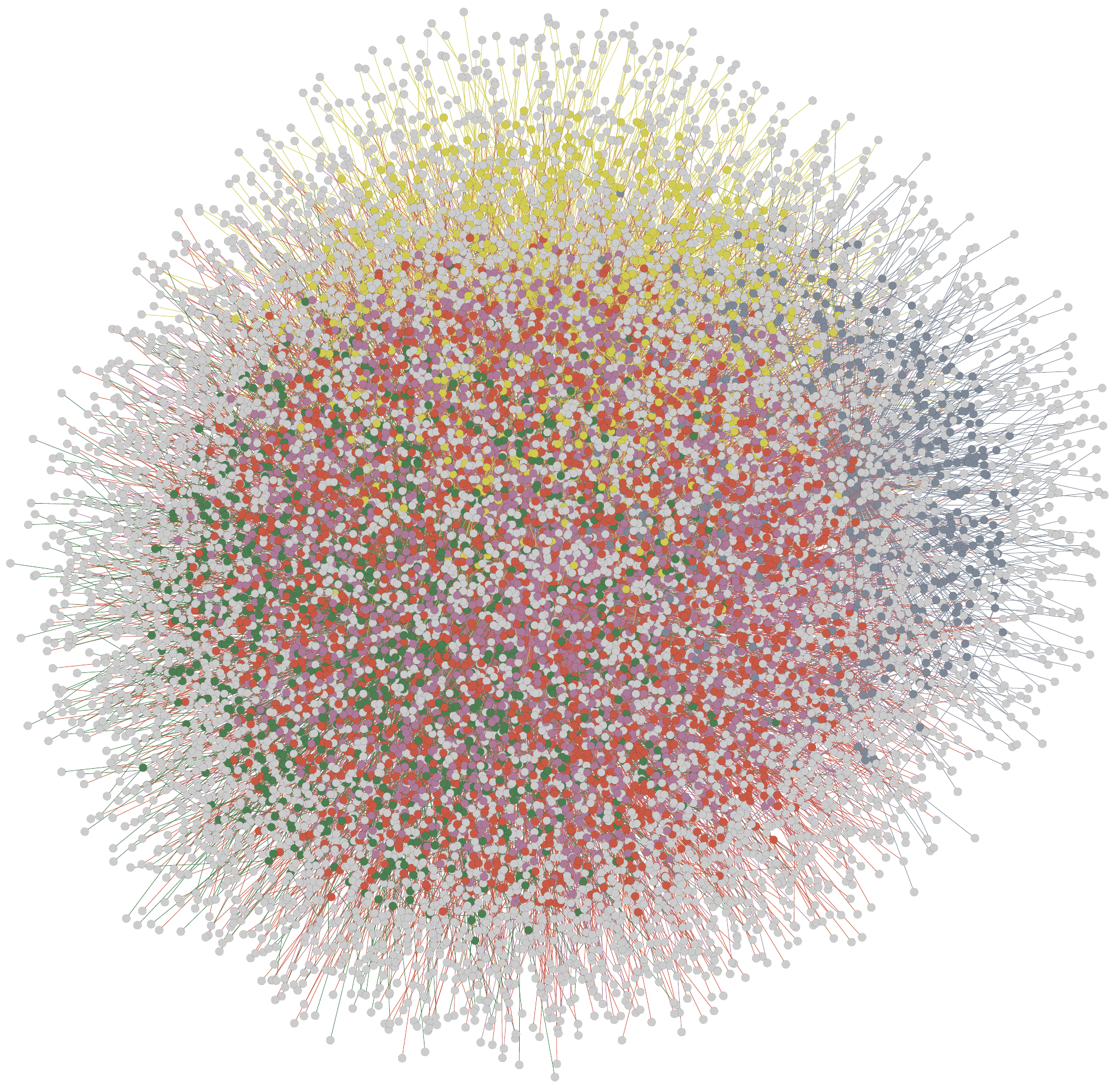}\label{fig:FactorizedRDFGraph}} 
\subfloat[{\bf Statistics} of Factorized RDF Graph]{
	\begin{tabular}[b]{|c|c|c|}
		\hline
       \multirow{1}{*}{\textbf{S\#}} & \multicolumn{1}{c|}{\multirow{1}{*}{\textbf{Parameter}}} & \multirow{1}{*}{\textbf{Value}}   \\ 
        \hline
        \multirow{1}{*}{$1$}&\multirow{1}{*}{Connected Components}&\multirow{1}{*}{$1.0$} \\   \hline
	    \multirow{1}{*}{$2$}&\multirow{1}{*}{Network Centralization}&\multirow{1}{*}{$0.1$} \\   \hline
	    \multirow{1}{*}{\textbf{3}}&\multirow{1}{*}{\textbf{Avg. \# of Neighbors}} &\multirow{1}{*}{\textbf{2.5}} \\   
	   \hline
	    \multirow{1}{*}{$4$}&\multirow{1}{*}{Network Density}&\multirow{1}{*}{$0.0$} \\   
    \hline
	    \multirow{1}{*}{5}&\multirow{1}{*}{Multi-edge Node Pairs}&\multirow{1}{*}{5.0}  \\    \hline
	    \multirow{1}{*}{$6$}&\multirow{1}{*}{Network Heterogeneity}&\multirow{1}{*}{$9.2$}\\      \hline
	\end{tabular}
\label{fig:statsFactorized}} 
\\
\vspace{5pt}
\subfloat[{\bf NT} Fact. vs Original ]{\includegraphics[width=0.6\linewidth]{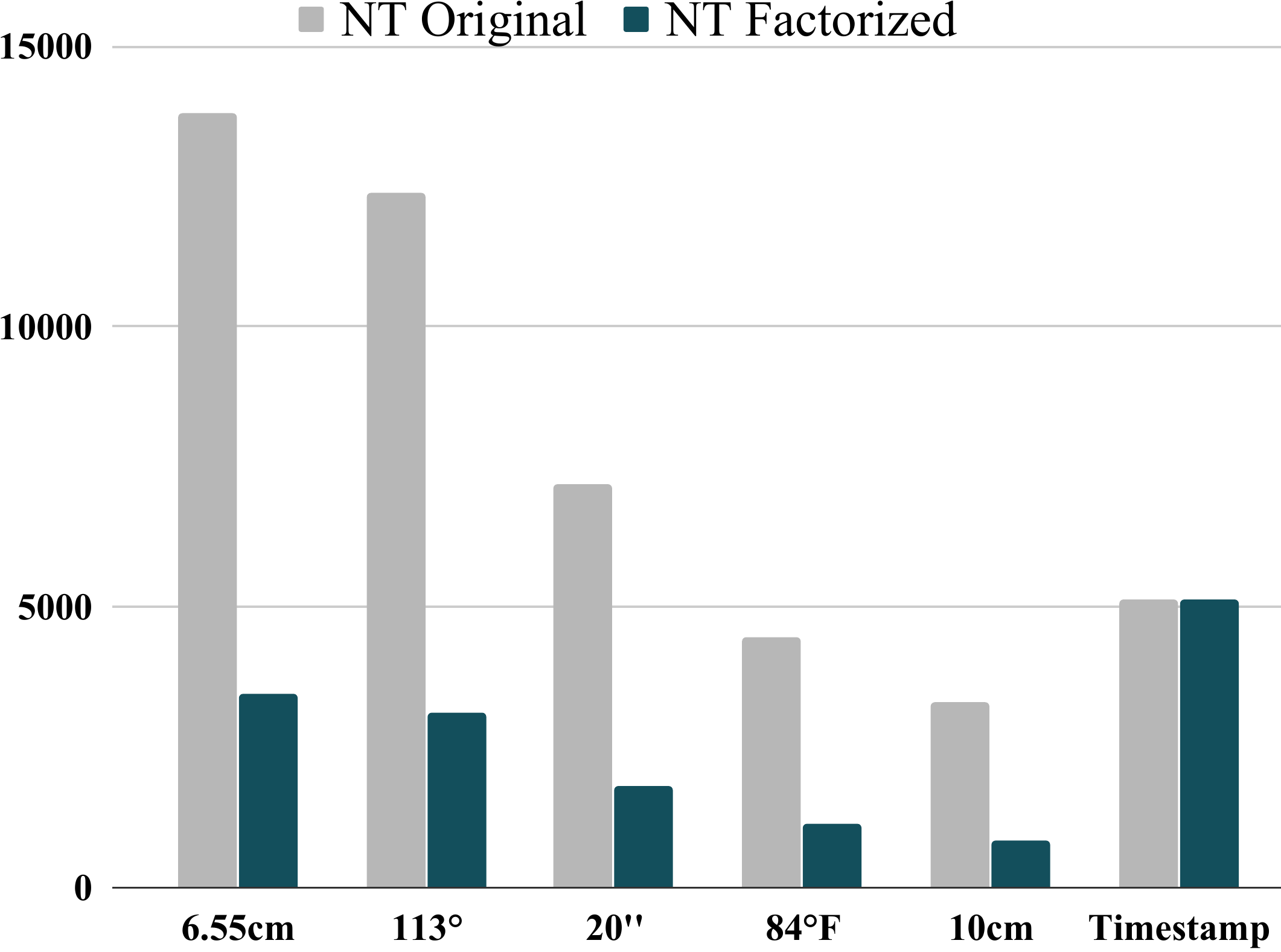}\label{fig:factorizedNT}}
\caption{{\bf Factorization of the Running Example.} Factorization reduces the number of RDF triples related to the same value. (a) Factorized (Fact.) RDF Graph of ~\autoref{fig:OriginalRDFGraph}; (b) Statistics of the fact. RDF graph; (c) The number of factorized triples. The graph and statistics are generated by the {\tt Cytoscape tool} (\url{http://www.cytoscape.org/}).}
\label{fig:motivatingExampleFactorization}
\end{figure*}

\autoref{fig:queryExecution} presents two SPARQL queries: \autoref{fig:queryOriginal} and \autoref{fig:queryFactorized} present an original query $Q$ and rewriting of $Q$ produced by Algorithm~\ref{algo:query}. 
\autoref{fig:transformations} shows the rewriting of SPARQL query in \autoref{fig:queryOriginal}. Rules \textit{fssn2}, \textit{fssn5}, \textit{fssn6}, and \textit{fssn7} from \autoref{tab:entailmentRules} are used to rewrite the query. The algorithm replaces each triple pattern in a BGP that instantiates the head of a rule in $SR$ by the body of the rule, e.g., the triple pattern ($?obs\;\textit{:procedure}\;\textit{:LGVI1}$) instantiates the head of rule \textit{fssn2}, thus, the triple pattern in the BGP is replaced with the body of \textit{fssn2}, as shown in \autoref{fig:transformations}. Moreover, the variables corresponding to the observations and measurements in the original query represent the surrogate observations and measurements in the rewritten query, consequently, these variables are replaced by the new variables in the query clauses. The variable substitution for observation $?obs$ by $?Xobs$ is maintained during the rewriting process using rule \textit{fssn2} in order to retrieve the original observations, if required. Similarly, other triple patterns in the BGP each matching the head of a rule, i.e., \textit{fssn5}, \textit{fssn6}, and \textit{fssn7}, are replaced by the body of the rule, and the variable substitutions of $?obs$ and $?m$ by $?Xobs$ and $?Xm$, respectively, are maintained for the query clauses.
The evaluation of both, original and rewritten, queries produce the same results. 
Another important property is that the time complexity of the original and rewritten queries is also the same.

\begin{thm}
Given  $G$ and $G'$ such that $G'$ is a factorized RDF graph of $G$.   Let $Q$ and $Q'$ be SPARQL queries where $Q'$ is a rewritten query of $Q$ over $G'$ generated by Algorithm~\ref{algo:query}. The problem of evaluating $Q'$ against $G'$ is in: (1) PTIME if query $Q$ has only AND and FILTER operators; (2) NP-complete if query $Q$ has expressions with AND, FILTER, and UNION operators; and (3) PSPACE-complete for OPTIONAL graph pattern expressions.
\label{teo:complexity}
\end{thm}
\begin{proof}
We proceed with a proof by contradiction. Assume that complexity of $Q'$ is higher than $Q$. Then, UNION  or OPTIONAL operators not included in $Q$ are added to $Q'$. However,  Algorithm~\ref{algo:query} only changes triple patterns over $G$ by triple patterns against $G'$. Additionally, Algorithm~\ref{algo:query} includes new JOINs (AND operator). However, adding AND or FILTER operators does not affect the complexity of the problem of evaluating $Q'$ over $G'$, and contradicting the fact that the complexity of $Q'$ is higher than $Q$.  
\end{proof}

\begin{figure*}[tb]
\centering
   \subfloat[Universal Parquet Table for Observations]{
      \includegraphics[width=0.7\textwidth]{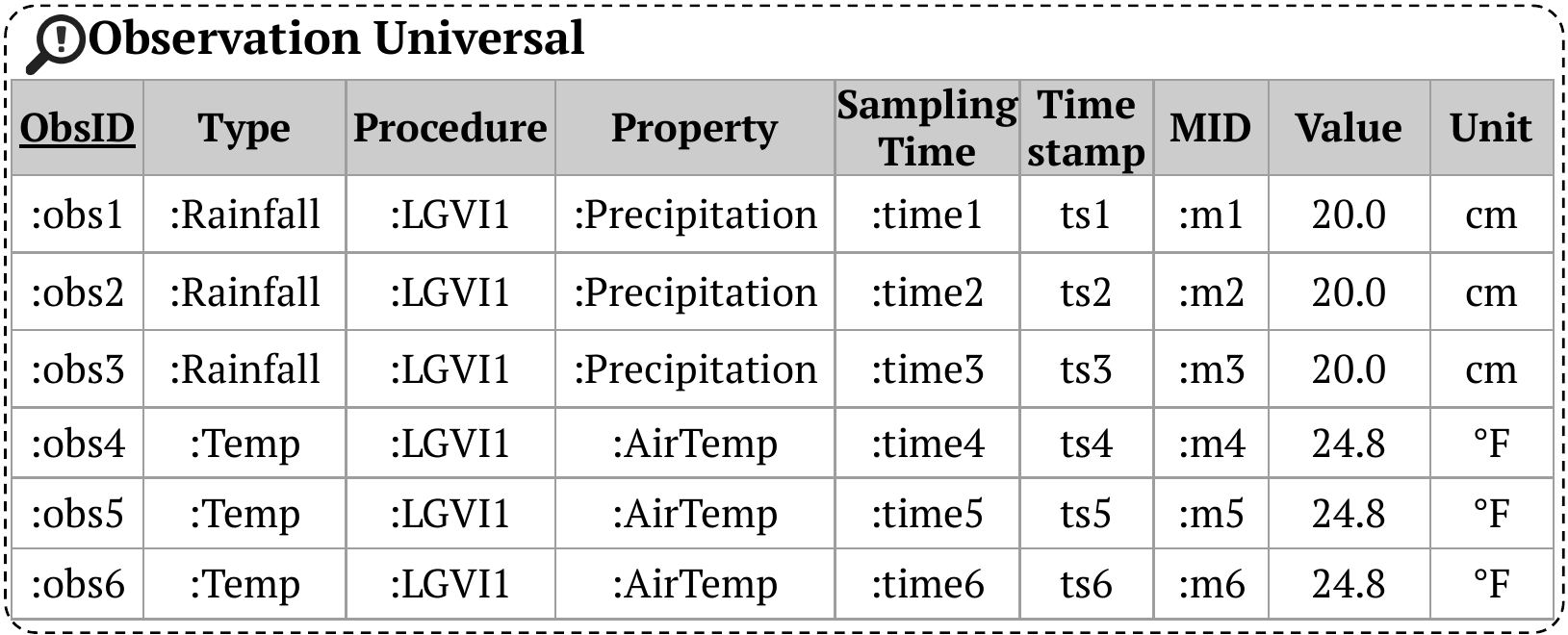}
      \label{fig:Universal}}
    \\
    \vspace{5pt}
    \subfloat[Factorized Data Parquet Tables]{
      \includegraphics[width=0.7\linewidth]{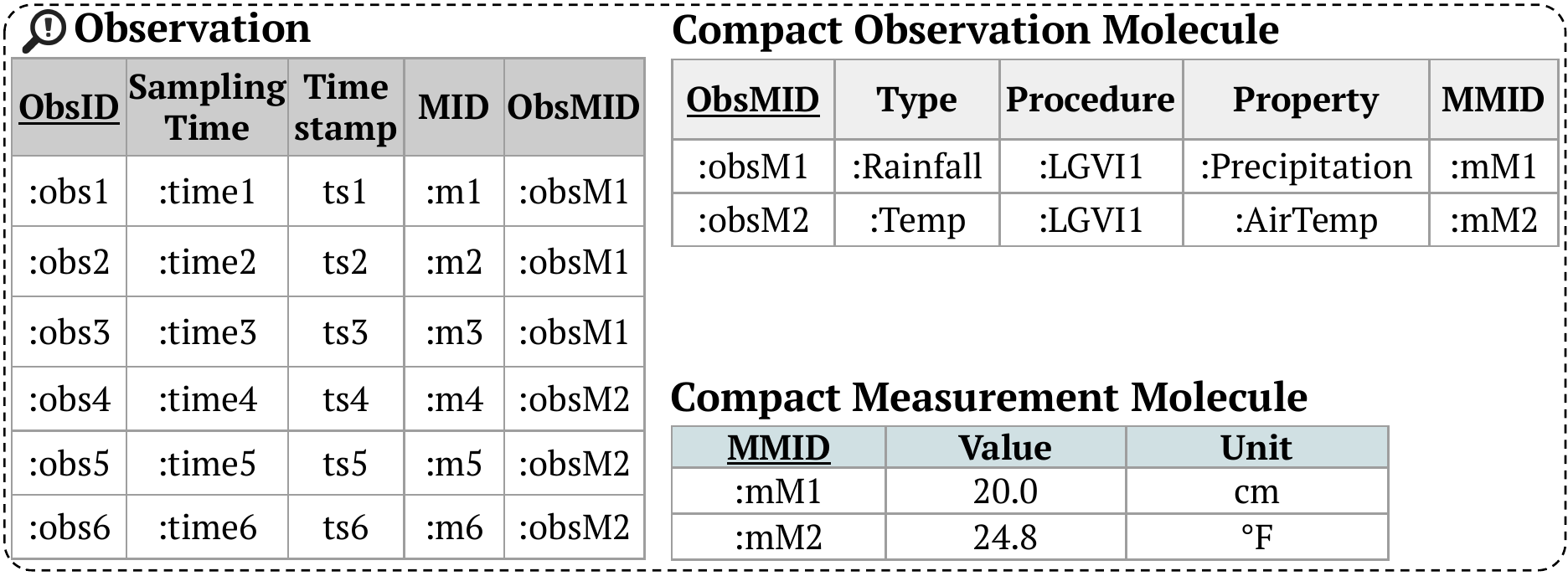}
      \label{fig:FactorizedTables}}
   \caption{{\bf Factorized Tabular Representation of RDF Graphs}. Parquet tables are utilized to represent RDF graphs in Spark.  (a) A universal table stores all the data of the original RDF graph.(b) Factorized data is represented in three parquet tables to store compact observation and measurement molecules.}
            \label{fig:factorizedTabular}
\end{figure*}

\section{Tabular Representation of RDF Graphs}
\label{sec:tabular}
Sensor data tend to stack up quickly, scaling up to large amounts of data. 
In order to capture that growth, we opt for representing factorized data in tabular format, so that Big Data processing technologies can be used.
For that purpose, we choose to store the data in the modern, columnar-oriented \emph{ Parquet}\footnote{\url{https://parquet.apache.org/}} storage format. 
We propose tabular representations of both the original and factorized RDF graphs (in \autoref{fig:multiplicityExp} and \autoref{fig:factorizedGraph}, respectively), shown in \autoref{fig:factorizedTabular} and \autoref{fig:RDFMTTabular}. %Parquet uses Run-Length Encoding (RL), whereby repeated numerical values are encoded into pairs of the value and its occurrences number which allows an efficient representation of large datasets of semantic sensor data. Moreover, 
Columnar nature of Parquet makes it best suited for scenarios where queries access only a few number of columns from a \textit{wide} table of many columns. Parquet pulls only the requested columns, contrary to row-oriented storage. %, where the whole row is read even if only few columns are requested.
%In other terms, whenever the ratio between the number of attributes used in a query and the number of attributes of the accessed table is small. 
We rely on these properties of Parquet tables, and represent RDF graphs using a \emph{universal} table.   
The universal tabular representation, {\tt Observation Universal} in \autoref{fig:Universal}, of original RDF graph in \autoref{fig:multiplicityExp},  contains all the properties of an observation, i.e., {\tt rdf:type}, {\tt:procedure}, {\tt:property}, {\tt :result}, {\tt :samplingTime}, {\tt :value}, {\tt :unit}, and {\tt :timestamp}.
These predicates are modeled with the attributes: {\tt Type}, {\tt Procedure}, {\tt Property}, {\tt MID}, {\tt SamplingTime}, {\tt Value}, {\tt Unit}, and {\tt Timestamp}, respectively.
The tabular representation of the factorized RDF graph in \autoref{fig:factorizedGraph} is shown in \autoref{fig:FactorizedTables}.
The {\tt Compact Observation Molecule} table models the properties {\tt rdf:type}, {\tt:procedure}, {\tt:property}, and {\tt :result}  of a surrogate observation with the attributes {\tt Type}, {\tt Procedure}, {\tt Property}, and {\tt MMID}, respectively.
The {\tt Compact Measurement Molecule} table contains the properties {\tt:value} and {\tt:unit} describing a surrogate measurement. Note that the type {\tt :MeasureData} is implicitly represented in the table name. The {\tt Observation} factorized table contains the observation predicates that are not represented in the {\tt Compact Observation Molecule} and {\tt Compact Measurement Molecule} tables, as well as a reference to the corresponding surrogate observations, as a foreign key.
%For example, the predicates {\tt :result}, {\tt :samplingTime}, and {\tt :timestamp} are not included in the {\tt Compact Observation Molecule} and {\tt Compact Measurem-} {\tt ent Molecule} tables, but they are represented in the {\tt Observation} table by the attributes {\tt MID}, {\tt SamplingTime}, and {\tt Timestamp}, respectively.
%Moreover, an association between an observation and related surrogate observation, described using {\tt :observationOf} property in the factorized RDF graph, is represented by the foreign key {\tt ObsMID}.
Furthermore, SPARQL queries against original and factorized graphs are translated into SQL queries over universal and factorized tables, respectively. 
\autoref{fig:UniversalTabularQuery} shows SQL representations of SPARQL queries in 
 \autoref{fig:queryExecution}. The evaluation of the SQL queries against the universal and factorized tables is the same as the SPARQL queries over the RDF graphs. 
\begin{figure*}[ht!]
\centering
   \vspace{0pt}\subfloat[Query Universal Table]{
      \includegraphics[width=0.6\textwidth]{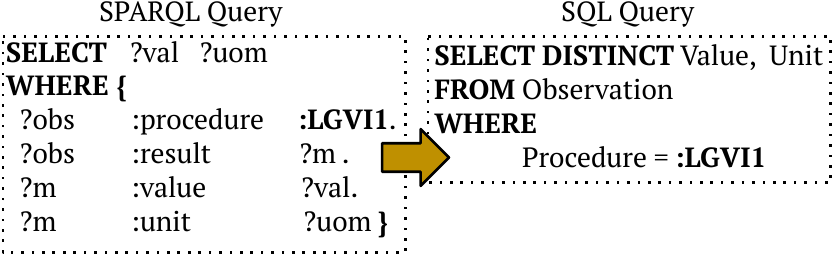}
      \label{fig:queryUniversal}}
      \\
   \vspace{5pt}\subfloat[Query Factorized Data Tables]{
      \includegraphics[width=0.6\linewidth]{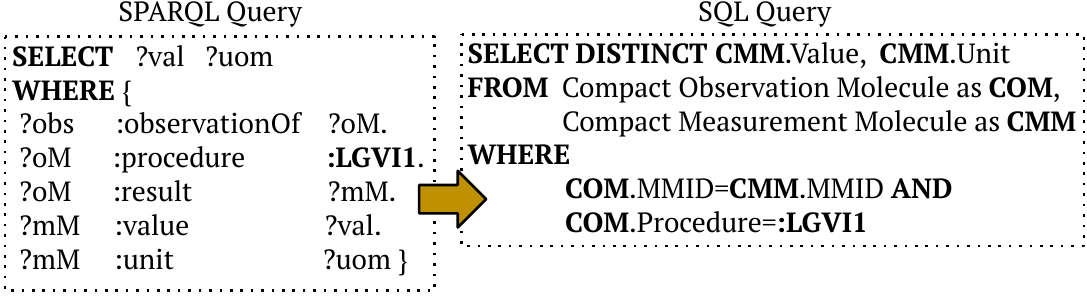}
      \label{fig:queryFactorizedTables}}
  \caption{{\bf Query Evaluation Over Universal and Factorized Tables.} SPARQL queries over original and factorized RDF graphs and their corresponding SQL queries are presented. (a) SQL query over the universal parquet table; (b) SQL query against the parquet tables representing the factorized RDF graph.}
			\label{fig:UniversalTabularQuery}
\end{figure*} 
Instead of using the universal tabular representations,  RDF graphs can be represented using the Class Template (CT) based tabular representations. 
For each CT around a class one table is created containing the properties of the class as attributes. Similarly, for each intra- or inter-link between the classes a binary table is created containing the identifiers from the corresponding CT tables.
Figure~\ref{fig:rdfmt} illustrates the CT-based tabular representations around the \texttt{:RainfallObs}, \texttt{:TempObs}, \texttt{:Instant}, and \texttt{:MeasureData} classes in Figure~\ref{fig:multiplicityExp}. 
\begin{figure*}[t]
\centering
   \subfloat[CT Based Parquet Table for Observations]{
      \includegraphics[width=0.7\textwidth]{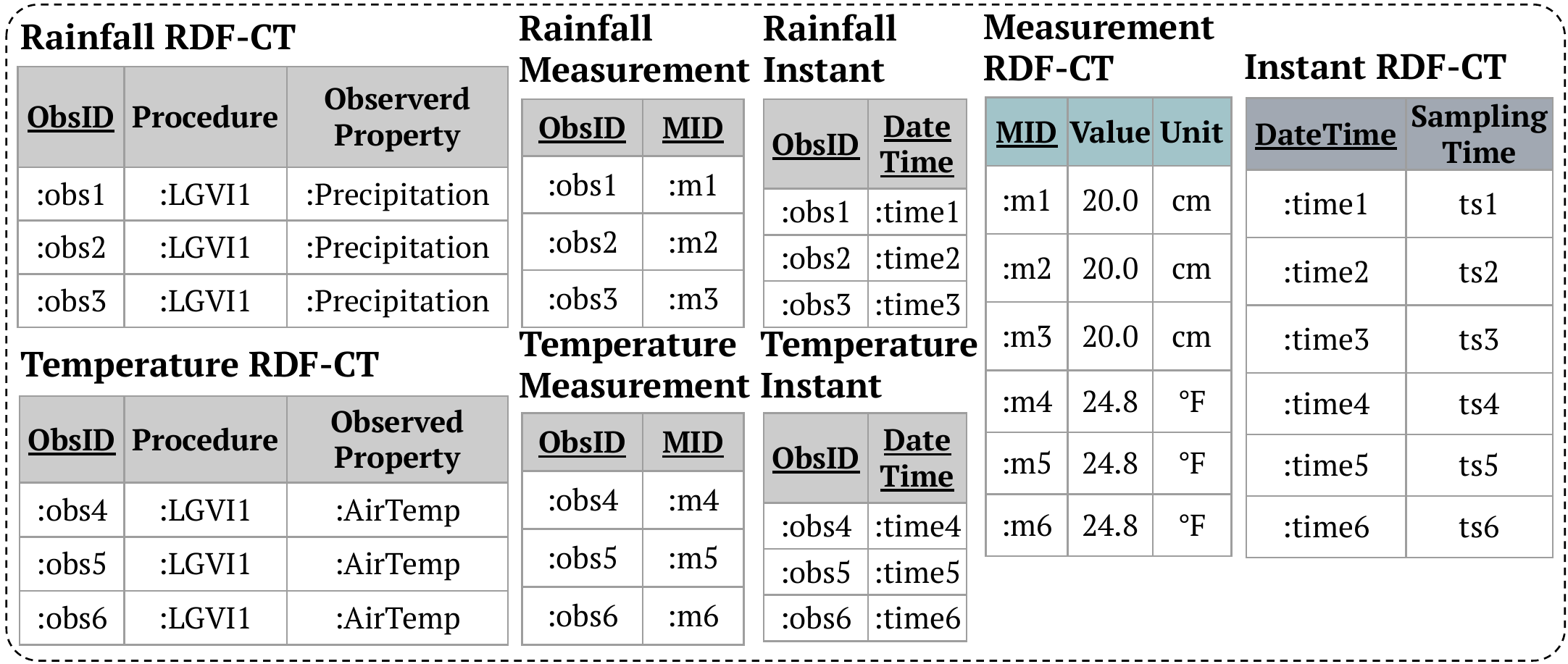}
      \label{fig:rdfmt}}%
      \\
    \vspace{5pt}\subfloat[Factorized CT Based Parquet Tables]{
      \includegraphics[width=0.7\linewidth]{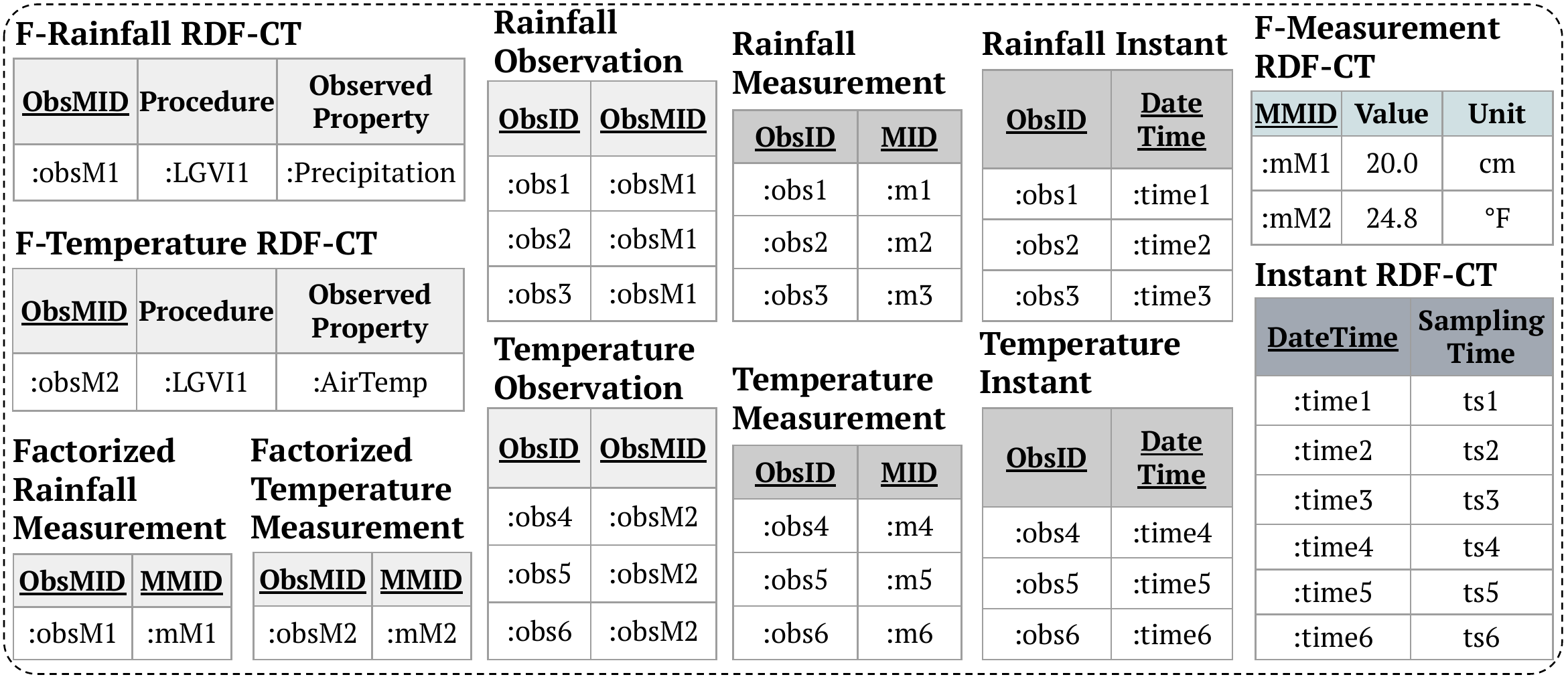}
      \label{fig:Factorizedrdfmt}}
   \caption{{\bf CT based Tabular Representation of RDF Graphs}. Parquet tables are utilized to represent CT-based tabular representations of RDF graphs in Spark. (a) Each CT-based table stores a class template collected from the original RDF graph.(b) Factorized RDF graph is represented in compact CT based tables. }
			\label{fig:RDFMTTabular}
\end{figure*}
The class templates of \texttt{:RainfallObs} and \texttt{:TempObs} are represented in \texttt{Rainfall CT} and \texttt{Temperature CT} with the attributes \texttt{Procedure} and \texttt{Property}. Similarly, \texttt{:MeasureData} and the properties \texttt{:value} and \texttt{:unit} are represented in \texttt{Measurement CT} with the attributes \texttt{Value} and \texttt{Unit}, respectively. \texttt{Instant CT} represents \texttt{:Instant} by modeling \texttt{:timestamp} property as \texttt{Timestamp}.
\texttt{Rainfall Measurement} models the association between the \texttt{:RainfallObs} and \texttt{:MeasureData} using the primary keys, \texttt{ObsID} and \texttt{MID}, from the corresponding CT-based tabular representations. Also, the association between \texttt{:RainfallObs} and \texttt{:Instant} is presented in \texttt{Rainfall Instant}. Similarly, association of \texttt{:TempObs} with \texttt{:MeasureData} and \texttt{:Instant} is presented in \texttt{Temperature Measurement} and \texttt{Temperature Instant}, respectively.
\begin{figure*}[h]
\centering
 
   \vspace{0pt}\subfloat[Query Class Template (CT) based Tables]{
      \includegraphics[width=0.7\textwidth]{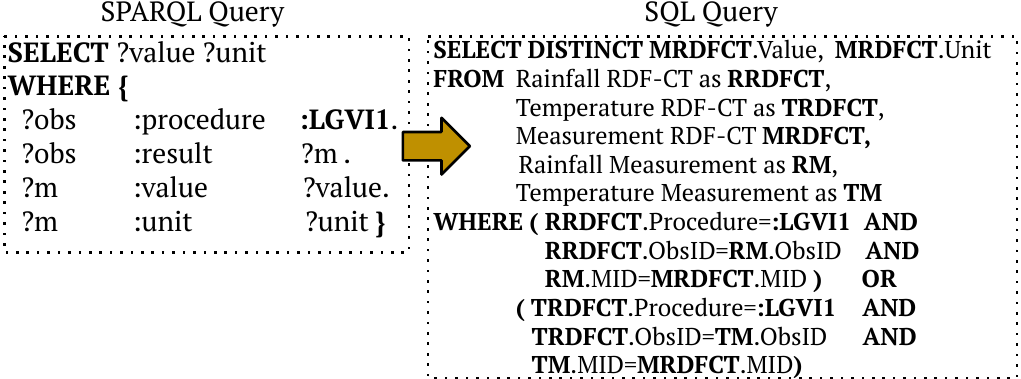}
      \label{fig:queryORDFMT}}
      \\
   \vspace{5pt}\subfloat[Query Factorized CT based Tables]{
      \includegraphics[width=0.7\linewidth]{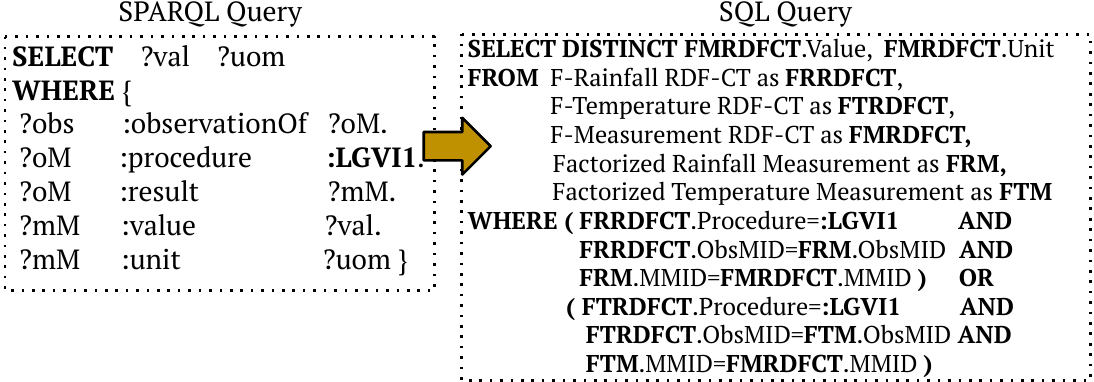}
      \label{fig:queryFRDFMT}}
  \caption{{\bf Query Evaluation Over CT-based Tables}. Original and rewritten SPARQL queries and the corresponding SQL queries against CT-based tables are presented. (a) Query over CT-based tables of the original RDF graph. (b) The SQL query over CT-based tabular representation of the factorized RDF graph.}
			\label{fig:queryRDFMT}
\end{figure*}

The CT-based tabular representations of the factorized RDF graph, in Figure~\ref{fig:factorizedGraph}, are shown in Figure~\ref{fig:Factorizedrdfmt}. \texttt{F-Rainfall CT} models the properties \texttt{:procedure} and \texttt{:property}, describing the surrogate rainfall observations, with the attributes \texttt{Procedure} and \texttt{Property}, respectively. 
Similarly, the CTs of the surrogate temperature observations are modeled in  \texttt{F-Temperature CT} with attributes \texttt{Procedure} and \texttt{Property}.
The surrogate measurements are modeled in the \texttt{F-Measurement CT} using \texttt{Value} and \texttt{Unit}. \texttt{Instant CT} models \texttt{:timestamp} property of \texttt{:Instant} with \texttt{Timestamp}. 
The links between the surrogate observations and measurements are represented in \texttt{Factorized Rainfall Measurement} and \texttt{Factorized Temperature Measurement}. Moreover, the explicit mappings between the original and surrogate rainfall observations are represented in \texttt{Rainfall Observation}. Similarly, \texttt{Temperature Observation} stores the mappings between the original and surrogate temperature observations. In addition, \texttt{Rainfall Measurement} and \texttt{Temperature Measurement} represent association of the original rainfall and temperature observations, respectively, with the corresponding measurements.  Furthermore, the links of the original rainfall and temperature observations with the corresponding timestamps are represented in \texttt{Rainfall Instant} and \texttt{Temperature Instant}, respectively. 
\autoref{fig:queryRDFMT} illustrates the CT based SQL representations of the SPARQL queries in \autoref{fig:queryExecution}. The results of the SQL queries against CT based tabular representations of the original and factorized RDF graphs are the same as the SPARQL queries over the original and factorized RDF graphs.

\begin{thm}
The decomposition of the {\tt Observation} universal table into factorized tables: {\tt Observation}, {\tt Compact Observation Molecule}, and {\tt Compact  Measurement Mole-} {\tt cule}, is \emph{loss-less join}. 
\label{teo:loss-less-universal}
\end{thm}
\begin{proof}
Considering the following functional dependencies hold in the universal and factorized tables:
\begin{itemize}
\item {\tt ObsMID} $\rightarrow$ {\tt Type}, {\tt Procedure}, {\tt Property}, {\tt MMID}
\item {\tt MMID} $\rightarrow$ {\tt Value}, {\tt Unit}
\item {\tt ObsID} $\rightarrow$ {\tt SamplingTime}, {\tt Timestamp}, {\tt MID}, {\tt ObsMID}
\end{itemize}
We can prove using the algorithm\cite{ullman20850principles} that the factorized tables are a \emph{loss-less join} decomposition of universal table $T$ that includes all the attributes in the {\tt Observation} universal plus {\tt ObsMID} and {\tt MMID}. The attributes of the {\tt Observation} universal can be projected from $G'$, thus, satisfying the \emph{loss-less join} condition.  
\end{proof}

\begin{thm}
If $G$ is an SSN RDF graph and $G'$ is a factorized RDF graph of $G$, and $T_1$ is the factorized tabular representation of $G'$, then $T_1$ is in third normal form with respect to the universal representation of $G$.
\end{thm}
\begin{proof}
Recall \cite{codd1972further}, a table is in third normal form if for every $X \rightarrow Y$
\begin{multicols}{2}
\begin{itemize}
    \item $X$ is a super key, or
    \item $Y-X$ is a prime attribute
\end{itemize}
\end{multicols}
Considering that the following functional dependencies hold in both the universal, and factorized tables:
\begin{itemize}
\item {\tt MMID} $\rightarrow$ {\tt Value}, {\tt Unit}
\item {\tt ObsMID} $\rightarrow$ {\tt Type}, {\tt Procedure}, {\tt Property}, {\tt MMID}
\item {\tt ObsID} $\rightarrow$ {\tt SamplingTime}, {\tt Timestamp}, {\tt MID}, {\tt ObsMID}
\end{itemize}
It can be demonstrated that all the tables created after factorization are in 3NF. 
\end{proof}

\begin{thm}
The decomposition of the {\tt Class Template (CT)} based tables representing sensor data into the {\tt factorized CT} based tables is \emph{loss-less join}. 
\label{teo:loss-less-temp}
\end{thm}
\begin{proof}
Consider the following functional dependencies hold in {\tt CT} and {\tt factorized CT} tables:

\begin{itemize}
\item {\tt ObsMID} $\rightarrow$ {\tt Procedure}, {\tt Property}
\item {\tt MMID} $\rightarrow$ {\tt Value}, {\tt Unit}
\item {\tt ObsMID}, {\tt MMID} $\rightarrow$ {\tt ObsMID}, {\tt MMID}
\item {\tt ObsID}, {\tt ObsMID} $\rightarrow$ {\tt ObsID}, {\tt ObsMID}
\item {\tt ObsID}, {\tt MID} $\rightarrow$ {\tt ObsID}, {\tt MID}
\item {\tt ObsID}, {\tt SamplingTime} $\rightarrow$ {\tt ObsID}, {\tt SamplingTime}
\item {\tt SamplingTime} $\rightarrow$ {\tt Timestamp}
\end{itemize}

We can prove using the algorithm\cite{ullman20850principles} that the {\tt factorized CT} based tables are a \emph{loss-less join} decomposition of the {\tt CT} based tables that includes all the attributes in the {\tt CT} tables plus {\tt ObsMID} and {\tt MMID}. The attributes of the {\tt CT} tables can be projected from $G'$, thus, satisfying the \emph{loss-less join} condition.  
\end{proof}

\begin{thm}
If $G$ is an SSN RDF graph and $G'$ is a factorized RDF graph of $G$, and $T_2$ is the {\tt Class Template (CT)} based tabular representation of $G'$, then $T_2$ is in third normal form with respect to the {\tt CT} based tabular  representation of $G$.
\end{thm}
\begin{proof}
Recall \cite{codd1972further}, a table is in third normal form if for every $X \rightarrow Y$ 
\begin{itemize}
    \item $X$ is a super key, or
    \item $Y-X$ is a prime attribute
\end{itemize}

Considering the following functional dependencies hold in {\tt CT} based tables:

\begin{itemize}
\item {\tt ObsMID} $\rightarrow$ {\tt Procedure}, {\tt Property}
\item {\tt MMID} $\rightarrow$ {\tt Value}, {\tt Unit}
\item {\tt ObsMID}, {\tt MMID} $\rightarrow$ {\tt ObsMID}, {\tt MMID}
\item {\tt ObsID}, {\tt ObsMID} $\rightarrow$ {\tt ObsID}, {\tt ObsMID}
\item {\tt ObsID}, {\tt MID} $\rightarrow$ {\tt ObsID}, {\tt MID}
\item {\tt ObsID}, {\tt SamplingTime} $\rightarrow$ {\tt ObsID}, {\tt SamplingTime}
\item {\tt SamplingTime} $\rightarrow$ {\tt Timestamp}
\end{itemize}

It can be demonstrated that all the factorized tables are in 3NF. 
\end{proof}

\begin{table*}[t] 
			\centering 
			%\scriptsize
			\caption{{\bf Datasets}: Description of the semantic sensor datasets; weather and smart city datasets; collected from the United States and Aarhus, Denmark, respectively.}
			%\resizebox{\columnwidth}{!}{%
			\begin{tabular}{|c | c | c | c |c | c | c |}
				\hline
				\multicolumn{4}{|c}{\textbf{Weather Dataset}} & \multicolumn{3}{|c|}{\textbf{Smart City Dataset}} \\ \hline
				\textbf{ID} &\textbf{Climate Event}  & \textbf{\#Triples} &\textbf{\# Obs} & \textbf{ID}   & \textbf{\#Triples} &\textbf{\# Obs}   \\ \hline
				D1 & Blizzard  & ~38,054,493 & ~~4,092,492 & C1 & 47,487,800 & 4,748,884 \\
				D2 & Hurricane Charley  & 108,644,568 &  11,648,607 & C2 & 47,051,850 & 4,705,267\\
				D3 & Hurricane Katrina  & 179,128,407 & 19,233,458 & C3 & 56,816,196 & 5,681,712\\ \hline
			\end{tabular}
			%}
			\label{tab:table1}    
		\end{table*}

\section{Experimental Study}
\label{sec:eval}
We empirically study the effect of the proposed factorization techniques over RDF implementations accessible through RDF and Big Data engines. We evaluate the impact on the size of the factorized RDF graphs as well as on query execution time in different query engines.
RDF-3X \cite{NeumannW10}  is utilized to evaluate the influence of the proposed techniques over the RDF stores. Spark~\cite{DBLP:journals/cacm/ZahariaXWDADMRV16} is used to study the tabular representation of RDF graphs.
%Federated query engines, like MULDER\cite{endris2017mulder} and ANAPSID\cite{AcostaVLCR11}, allow for accessing RDF data available through Web interfaces like SPARQL endpoints. They are equipped with physical implementations for the operators of the relational algebra and are able to produce results incrementally. MULDER and ANAPSID are used to assess the impact of the factorization approach over the engines utilizing RDF implementations available through the endpoints.
In this work, we investigated the following research questions:
%We empirically assessed the following research questions:
\begin{inparaenum}[\bf {\bf RQ}1\upshape)]
     \item Are the proposed factorization techniques able to reduce the size of the semantically represented sensor data?
    \item How is the factorization time affected by the size of the RDF graphs? %same as conference paper
    \item What is the impact of the queries against factorized RDF graphs over the query execution time?
    \item Is the performance of queries against factorized RDF graphs affected by the size of the factorized RDF graphs or RDF implementation?%same as conference paper
\end{inparaenum}
\begin{table*}[t]
\centering
\caption{{\bf Effectiveness of the Semantic Sensor Data Factorization}. Number of triples ({\bf NT}) before and after factorization along with \%age NT savings. %Savings in the number of triples ({\bf \%age NT Savings}) increases as the size of the dataset, while average number of triples per observation {\bf Avg. NT per Obs.} decreases.
}
	\begin{tabular}{| l | c | c | c | c | c |}
		\hline
	    \multicolumn{1}{|c|}{\textbf{Dataset}}	&  \multicolumn{2}{c|}{\textbf{Number of Triples(NT)}}  & \textbf{\%age NT } & \multicolumn{2}{c|}{\textbf{Avg. NT per Obs. }} \\ \cline{2-3} \cline{5-6} 
	    \multicolumn{1}{|c|}{\textbf{ID}}	& \textbf{Original} & \textbf{Factorized}  & \textbf{Savings} & \textbf{Original}& \textbf{Factorized} \\
        \hline
		\textbf{D1}	& ~38,054,493 & ~17,800,156 &  53.22 & 9.29 & 4.34  \\   \hline	
		\textbf{D1D2} & 146,699,061 & ~63,993,774 &  56.38 &9.32& 4.06  \\  \hline
		\textbf{D1D2D3} & 325,827,468 &  136,979,696 & {\bf 57.96}& 9.31& {\bf 3.92} \\  \thickhline
			\textbf{C1}	& ~47,487,800 & ~23,937,396 &  49.59 & 9.99 & 5.04  \\   \hline	
		\textbf{C1C2} & ~94,539,650 & ~47,621,691 &  49.63 &9.99& 5.04  \\  \hline
		\textbf{C1C2C3} & 151,355,846 &  ~76,223,192 & {\bf 49.64}& 9.99& {\bf 5.04} \\  \hline
	\end{tabular}
\label{tbl:NT}
%\label{fig:datatsets}
\end{table*}
\begin{table*}[t]
\centering
\caption{{\bf Efficiency of the Semantic Sensor Data Factorization}. Time that elapses during factorization ({\bf FT}) as well as the RDF3X Loading Time ({\bf LT}). % for the factorized datasets is less than the Loading Time ({\bf LT}) for original datasets.
}
	\begin{tabular}{| l | c | c | c |}
		\hline
	    \multicolumn{1}{|c|}{\textbf{Dataset}} & \textbf{Factorization}  & \multicolumn{2}{c|}{\textbf{RDF3X LT(s)}} \\ 
	    \cline{3-4} 
	    \multicolumn{1}{|c|}{\textbf{ID}} & \textbf{Time FT(s)}  & \textbf{Original} & \textbf{Factorized} \\
        \hline
		\textbf{D1} & ~~417.229 & ~460.511 & ~252.976 \\   \hline	
		\textbf{D1D2} & 1,260.495 & 1,887.626 & ~970.150 \\  \hline
		\textbf{D1D2D3} & 2,147.239 & 3,822.723 & 1,982.697 \\  \hline
	\end{tabular}
\label{tbl:LT}
%\label{fig:datatsets}
\end{table*}
The experimental configuration to evaluate the research questions mentioned above is as follows:
~\\
\noindent
{\bf Datasets:}  Evaluation is conducted over two sensor datasets~\cite{ali2015citybench,patni2010linked} described using the Semantic Sensor Network (SSN) Ontology. 
The RDF datasets describing weather observations are collected from around 20,000 weather stations in the United States\footnote{Available at: \url{http://wiki.knoesis.org/index.php/LinkedSensorData}}.  Realtime smart city datasets are collected from the city of Aarhus, Denmark. The smart city datasets encompasses the traffic, pollution, and parking observations \footnote{Available at: \url{http://iot.ee.surrey.ac.uk:8080/datasets.html}}.
\autoref{tab:table1} describes the main characteristics of these RDF datasets. 
~\\
\noindent
{\bf Queries:}  The SRBench-Version 0.9 queries\footnote{\url{https://www.w3.org/wiki/SRBench}}
are used as baseline in our experimental testbed. 
Because RDF-3X does not evaluate queries with the OPTIONAL operator, query 2 is modified to include only one BGP. 
Also, the STREAM clause, ASK queries, aggregate modifiers like AVG, GROUP BY, and HAVING  are not supported. 
So, only SELECT queries without aggregate modifiers are part of our testbed.
Queries range from simple queries with one triple pattern to complex queries having up to 14 triple patterns with UNION and FILTER clauses 
\footnote{Details can be found at \url{https://sites.google.com/site/fssdexperimets/}}.
~\\
\noindent
{\bf Metrics:} We report on the following metrics: 
\begin{inparaenum}[\bf a\upshape)]
\item {\bf Number of Triples (NT)} in the semantic sensor data collection.
\item {\bf Percentage Savings (\%age NT Savings)} in the number of RDF triples after factorization; higher the better.
\item {\bf Average Number of Triples per Observation (avg. NT per Obs.)} represents the average number of RDF triples describing an observation; lower the better.
\item {\bf Factorization Time (FT)} is the elapsed time between the request of factorization and the generation of the factorized RDF graph.
\item {\bf RDF3x Loading Time (LT)} is the time required to load RDF data to RDF3x store.
{\bf FT} and {\bf LT} are computed as the {\it real time} of the {\it time} command of the Linux operating system.
\item {\bf Query Execution Time (ET)} is the elapsed time between the submission of the query to the engine and the complete output of the answer, and is measured as the {\it real time} produced by the {\it time} command of the Linux operation system. %, whereas, in RDF implementations accessible through endpoints, {\bf ET}, for SPARQL endpoints, is measured as the absolute wall-clock system time produced by the Python {\it time.time()} function. \item {\bf dief@t} measures the continuous behavior of a query engine that produces results incrementally. It computes the diefficiency of an engine in the first $t$ time units of the query execution~\cite{acosta2017diefficiency}. \item {\bf Time For the First Tuple (TFFT)} is the elapsed time spent by the approach to produce the first query answer. {\bf TFFT} is measured as the absolute wall-clock system time as reported by the Python {\it time.time()} function. \item {\bf Completeness (Comp)} is the percentage of the number of answers produced by the approach after executing a query. \item {\bf Throughput (T)} is the number of answers per second and is computed by dividing the total number of answers produced by the total execution time.
\end{inparaenum}
%For the RDF implementations available through endpoints, inverse of {\bf TFFT} and {\bf ET} are reported to have the same metric interpretation, i.e., higher is better.
~\\
\noindent
{\bf Implementation:}
Three series of experiments were conducted over the gradually integrating sensor datasets in Table~\ref{tab:table1}, i.e., D1, D1D2, and D1D2D3. \begin{inparaenum}[\bf i\upshape)] \item Algorithm~\ref{algo:factorization} is executed over the original RDF datasets to generate the factorized RDF representations. Moreover, original and factorized RDF datasets are loaded in RDF3X store. \item SPARQL queries are executed using RDF3X engine over original and factorized RDF datasets. %RDF3X engine executes queries over RDF data stored locally. 
The experiments are executed on a Linux Debian 8 machine with a CPU Intel I7 980X 3.3GHz and 32GB RAM 1333MHz DDR3.
Queries are run on both cold and warm cache.\footnote{To run cold cache, we clear the cache before running each query by performing the command  {\tt \scriptsize sh -c "sync ; echo 3 $>$ /proc/sys/vm/drop\_caches"}}  to assess the query performance when data is cached.
To run on warm cache, we executed the same query five times by dropping the cache just before running the first iteration of the query; thus, data temporally stored in cache during the execution of iteration $i$ can be used in iteration $i+1$.
\item In the third series of experiments, SQL queries were run on cold and warm cache using \emph{Apache Spark}\footnote{\url{http://spark.apache.org/}} over the universal, factorized, original and factorized CT-based tabular representations. These tabular representations are stored using Parquet format in \emph{HDFS}\footnote{https://hadoop.apache.org/}. 
%For all relational representations in the queries were run with cold and warm cache.
The experiments were conducted on a spark cluster of one master and three worker nodes.% created using \emph{Docker}\footnote{https://www.docker.com/} containers, and the datasets are stored on a hadoop cluster containing one namenode and three datanodes created using docker containers.
The experiments are performed on a machine with Intel(R) Xeon(R) Platinum 8160 CPU 2.10GHz and 23 RAM slots, where each RAM slot is DDR4 type, 32GB RAM size, and 2666MHz speed. The source code of the factorization approach is available on github\footnote{\url{https://github.com/SDM-TIB/SemanticSensorDataFactorization}}.
%Further, queries are run on cold and warm cache.
%\item In the third series of experiments MULDER and ANAPSID engines are used to access the RDF datasets available as a SPARQL endpoint using {\it Virtuoso 7.2.2}, where each original and factorized dataset resides in a dedicated Virtuoso docker container. All queries are executed using MULDER and ANAPSID over the original and the factorized datasets. The experiments are conducted on a machine with Intel(R) Xeon(R) Platinum 8160 CPU 2.10GHz and 23 RAM slots, where each RAM slot is DDR4 type, 32GB RAM size, and 2666MHz speed.
\end{inparaenum}
\begin{figure*}[tb]
\centering
\includegraphics[width=1.0\textwidth]{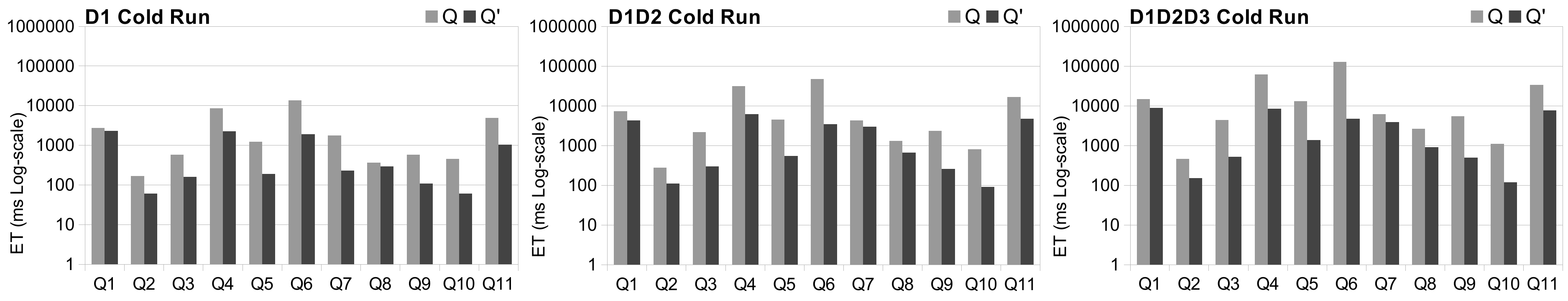}
\caption{{\bf Query Execution Time ET (ms Log-scale) over RDF3x}. Original SPARQL queries $Q$  and rewritten SPARQL queries  $Q'$  are evaluated on \textbf{cold} cache against original and factorized RDF graphs, respectively. Rewritten queries reduce execution time on factorized RDF graphs by up one order of magnitude. }
\label{fig:rdf3xCold}
\end{figure*}

\begin{figure*}[tb]
\centering
\includegraphics[width=1.0\textwidth]{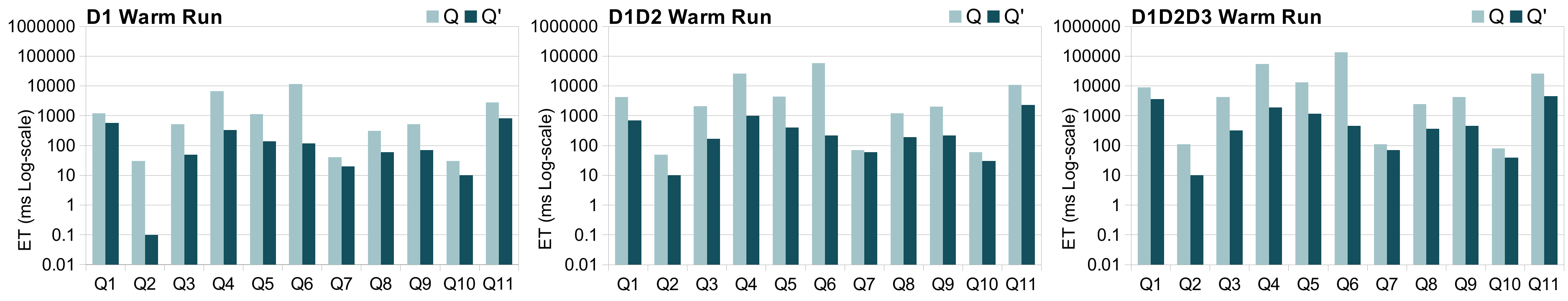}
\caption{{\bf Query Execution Time ET (ms Log-scale) over RDF3x}. SPARQL queries $Q$  and rewritten queries $Q'$  are evaluated on \textbf{warm} cache against original and factorized RDF graphs, respectively. To warm cache up, memory is flushed. %, and each query is executed five times; the lowest value of execution time is reported. 
The rewritten queries reduce query execution time by up two order of magnitude.}
\label{fig:rdf3xWarm}
\end{figure*}
~\\
\noindent
{\bf Efficiency and Effectiveness of Factorized RDF.}
For evaluating the efficiency and effectiveness of the proposed factorization techniques and to answer the research questions {\bf RQ1} and {\bf RQ2}, we execute \autoref{algo:factorization} by gradually integrating the datasets in Table~\ref{tab:table1}, i.e., {\bf D1}, {\bf D1D2}, and {\bf D1D2D3}. 
Effectiveness is reported based on the reduction of RDF triples ({\bf NT}), while
efficiency is measured in terms of factorization time ({\bf FT}) and RDF3X loading time ({\bf LT}). 
\autoref{tbl:NT} reports on the number of RDF triples ({\bf NT}) in datasets {\bf D1}, {\bf D1D2}, and {\bf D1D2D3} before and after the factorization, as well as in datasets {\bf C1}, {\bf C1C2}, and {\bf C1C2C3}. 
The results demonstrate that the proposed factorization techniques are capable of reducing the RDF triples by at least {\bf 53.22\%} in the datasets of weather observations, and  {\bf 49.59\%} in smart city dataset. 
Moreover, the results report that the factorized representation of sensor observations requires in average a small number of RDF triples, e.g., five RDF triples instead of ten in the weather dataset, while preserving all the information within the original RDF graph. 
These results allows us to positively answer research question {\bf RQ1}, i.e., factorized RDF graphs effectively reduce the size of RDF graphs. 
We also measure factorization time and factorized RDF loading time to RDF3X, and compare to the time required by RDF3X to upload the original RDF graphs, in \autoref{tbl:LT}. Algorithm~\ref{algo:factorization} as well as factorized RDF loading to RDF3X requires less than {\bf 50\%} of the time consumed by RDF3X during original RDF data loading. Thus, with these results research question {\bf RQ2} can be also positively answered. 
\begin{figure*}[tb]
\centering
      \includegraphics[width=1.0\textwidth]{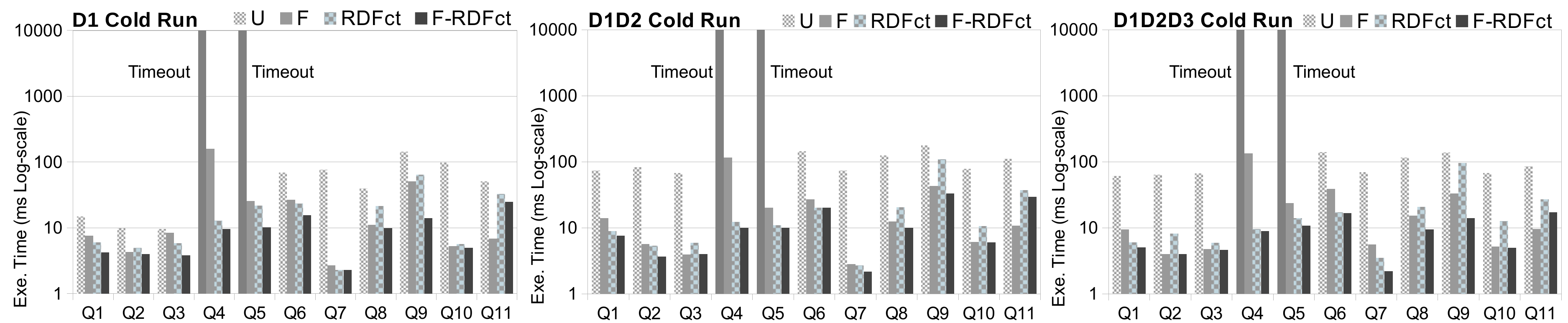}
   \caption{{\bf Query Execution Time ET (ms Log-scale) over Relations}. Query evaluation over tabular based representations in \textbf{cold} cache. Executions are timed out after 100 minutes. SQL version of the rewritten SPARQL queries over the factorized (\it{F}) and factorized CT tables (\it{F-RDFct}) reduce execution time.}
\label{fig:coldSpark}
\end{figure*}      
\begin{figure*}[ht!]
\centering
      \includegraphics[width=1.0\linewidth]{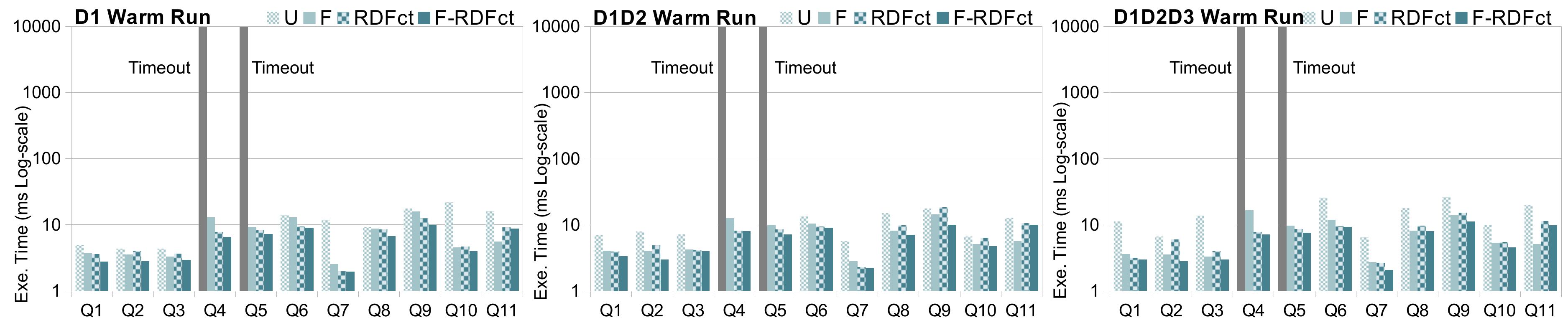}
   \caption{{\bf Query Execution Time ET (ms Log-scale) over Relations}. Query evaluation over tabular representations in \textbf{warm} cache. Execution timeout is 100 minutes. SQL queries execution time over the factorized (\it{F}) and factorized CT tables (\it{F-RDFct}) is less than the universal (\it{U}) and original CT tables (\it{RDFct}).}
   \label{fig:warmSpark}
\end{figure*}    
~\\
\noindent
{\bf Impact of Factorized RDF over Query Processing.}
We  analyze the efficiency of the proposed representations by running the queries generated using Algorithm~\ref{algo:query}. First, the impact of our approach on query execution is studied over centralized RDF engines; to evaluate the benefits of caching previous results, queries are executed on cold and warm caches.
The advantage of running these queries on cold and warm caches on RDF3X are analyzed over the gradually increasing original and factorized RDF datasets. 
The original queries {\tt Q} are compared to the reformulated queries {\tt Q'}. Original queries ({\tt Q}) are executed against the original datasets, while plans for reformulated queries ({\tt Q'}) are run against gradually increasing factorized datasets.    
Figure~\ref{fig:rdf3xCold} reports on the query execution time (milliseconds. log-scale) with cold cache, while
Figure~\ref{fig:rdf3xWarm} depicts the observed execution time when queries are run on warm cache; the minimum value is reported in all the queries. In all cases, reformulated queries over factorized RDF graphs exhibit better performance whenever they are run on cold and warm caches.
This observation supports the statement that because observation and measurement multiplicity is reduced to one in factorized RDF graphs, factorized queries produce small intermediate results which can be maintained in resident memory and re-used in further executions. 
Thus, the performance of factorized queries is considerable better with warm cache, overcoming other executions by up to three orders of magnitude, e.g., Q2 and Q6. 
Results also suggest that performance of reformulated queries is not affected by the RDF graph size, e.g., large RDF graphs like D1D2D3 with 325,827,468 RDF triples. 

We further analyse the effect of factorization when query processing is conducted over the relational representations of sensor data, i.e., universal and factorized tables, and the CT based tabular implementation of original and factorized RDF data. The performance of queries over Parquet tables depends on the number of attributes included in the query, as well as on the ratio between the attributes in the query and the attributes in the tables
\footnote{\url{https://parquet.apache.org/}}.
In queries against the universal table, the ratio between the number of attributes varies from {\bf 0.09} to {\bf 0.45}. While the ratio in factorized queries is in the range from {\bf 0.46} to {\bf 0.75}, and in original and factorized CTs is {\bf 0.25} and {\bf 0.78}. So, based on this statement, queries over the universal table should be faster than queries over the factorized tables and CT based tables. However, as observed in 
Figures~\ref{fig:coldSpark} and \ref{fig:warmSpark}, reformulated queries over factorized CT tables speed up execution time to almost two orders of magnitude, except Q11 where factorized tables are performing better. Factorized CTs reduce the size of tables by creating them around each molecule template and factorization further removes data redundancies. Actually, in queries Q4 and Q5, execution over the universal table times out after 100 minutes. These results indicate that the rewritten queries speed up query processing over big data engines. 
~\\
\noindent
{\bf Discussion.}
The presented experimental results confirm that the factorization techniques are able to reduce duplicated measurements in observational data without any information lost. Furthermore, since graphs can be factorized incrementally, savings are observed whenever new incoming tuples are related to measures previously collected. The benefits of the approach are reported in the reduction of the number of RDF triples of the factorized graphs, as well as in the execution time of queries rewritten over these factorized graphs. These savings are even more significant when the query engine provides efficient caching techniques to maintain in cache intermediate results of previously evaluated queries. Lastly, in the case of relational representation of factorized data in big data infrastructures, space savings are significant, enabling an efficient query execution over factorized tables.

\section{Conclusions and Future Work}
\label{sec:conclusion}
This article presents compact RDF representations for semantic sensor data to reduce data redundancy, while information is preserved and query execution performance is enhanced. Moreover, the effectiveness of the proposed approach was studied over several query engines.  Furthermore, tabular representations for a loss-less large-scale storage of factorized semantic sensor data are presented.
A factorization algorithm transforms original observations and measurements to a compact representation where data redundancy is reduced. 
Additionally, query rewriting rules and a query re-writing algorithm are presented. The query rewriting algorithm exploits the rewriting rules to rewrite SPARQL queries against factorized RDF graphs, and speeds up query execution time.
The factorized observations and measurements are also exploited to produce tabular representations for factorized RDF graphs utilizing Parquet tables.
We empirically evaluate the effectiveness of the proposed factorization techniques and results confirm that exploiting semantics encoded in semantic sensor data allow for reducing redundancy by up to 57.96\%, while the time taken by the process of factorizing RDF data is less than 50\% of loading time for the original RDF data in state-of-the-art RDF stores. Further, the loading time for factorized RDF data is reduced by more than 45\% of the loading time of original RDF data in native RDF stores. 
Also, we evaluated the impact of proposed compact representations over the diverse implementations available for RDF data, i.e., native RDF implementations and non-native large-scale tabular based implementations. 
Thus, \textit{CSSD} broadens the portfolio of tools that enable to semantically enrich sensor data. As the main limitation, \textit{CSSD} can only be applied to data coming from one single device. In the future, we plan to devise data integration techniques able to merge RDF molecules generated from the factorization of heterogeneous data collected either from sensors or static data sources of observational data.  
We will apply these techniques to the energy domain to facilitate the integration and analysis of data collected from diverse energy providers. 

\section*{Acknowledgments}
Farah Karim is supported by the German Academic Exchange Service (DAAD).

%\bibliographystyle{ACM-Reference-Format}

%\bibliographystyle{spmpsci.bst}
%\bibliography{sigproc} 

\vfill
\bibliographystyle{apalike}
{\small
\bibliography{sigproc}}

\vfill
\end{document}